\documentclass[fullpage,11pt]{article}
\usepackage{amsmath, amsthm, amssymb}
\usepackage{graphicx}
\usepackage{caption}
\usepackage{cite}
\usepackage{todonotes}
\usepackage[noend]{algorithmic}
\usepackage{wrapfig}

\usepackage[margin=1in]{geometry} 
\usepackage{hyperref}	% Handles web links

\theoremstyle{definition}
\newtheorem{definition}{Definition}
\theoremstyle{plain}
\newtheorem{theorem}{Theorem}
\newtheorem{lemma}{Lemma}

\newtheorem{corollary}{Corollary}
\newtheorem{conjecture}{Conjecture}

\usepackage{tikz}

\begin{document}
%\title{Dominated Minimal Separators are Tame While Nearly All Others are Feral}
\title{Dominated Minimal Separators are Tame \\ {(Nearly All Others are Feral)}}
%\title{Families with Dominated Minimal Separators are (Quasi-) Tame \\ {(Nearly All Others are Feral)}}
%\title{Dominating Minimal Separators with Few Vertices }
%\title{Weighted Independent Set on P$_k$-Free Graphs in Quasi-Polynomial Time}

\author{
Peter Gartland\thanks{University of California, Santa Barbara, USA. Emails: \texttt{petergartland@ucsb.edu}, \texttt{daniello@ucsb.edu}}
%\addtocounter{footnote}{-1}
%\and Daniel Lokshtanov\thanks{}
\and Daniel Lokshtanov\footnotemark[1]
%\thanks{University of California, Santa Barbara, USA. \texttt{daniello@ucsb.edu}}
}
\maketitle

\begin{abstract}
A vertex set $S$ in a graph $G$ is a {\em minimal separator} if there exist vertices $u$ and $v$ that are in distinct connected components of $G - S$, but in the same connected component of $G - S'$ for every $S' \subset S$. A class ${\cal F}$ of graphs is called {\em tame} if there exists a constant $k$ so that every graph in ${\cal F}$ on $n$ vertices contains at most $O(n^k)$ minimal separators. If there exists a constant $k$ so that every graph in ${\cal F}$ on $n$ vertices contains at most $O(n^{k \log n})$ minimal separators the class is {\em strongly-quasi-tame}. If there exists a constant $c > 1$ so that ${\cal F}$ contains $n$-vertex graphs with at least $c^n$ minimal separators for arbitrarily large $n$ then ${\cal F}$ is called {\em feral}. The classification of graph classes into tame or feral has numerous algorithmic consequences, and has recently received considerable attention.

A key graph-theoretic object in the quest for such a classification is the notion of a $k$-{\em creature}. 
A $k$-creature consists of $4$ disjoint vertex sets $A,B,X = \{x_1,\ldots,x_k\}, Y = \{y_1, \ldots y_k\}$ such that: (a) $A$ and $B$ are connected, (b) there are no edges from $A$ to $Y \cup B$ and no edges from $B$ to $X \cup A$, (c) $A$ dominates $X$ (every vertex in $X$ has a neighbor in $A$) and $B$ dominates $Y$ and (d) $x_iy_j$ is an edge if and only if $i = j$. It is easy to verify that a $k$-creature contains at least $2^k$ minimal separators. On the other hand, in a recent manuscript [Abrishami et al., Arxiv 2020] conjecture that every  hereditary class ${\cal F}$ that excludes $k$-creatures for some fixed constant $k$ is tame.

Our main result is a proof of a weaker form of the conjecture of Abrishami et al. More concretely, we prove that a hereditary class ${\cal F}$ is strongly quasi-tame if it excludes $k$-creatures for some fixed constant $k$ and additionally every minimal separator can be dominated by another fixed constant $k'$ number of vertices. The tools developed on the way lead to a number of additional results of independent interest. 

{\bf (i)} We obtain a complete classification of all hereditary graph classes defined by a finite set of forbidden induced subgraphs into strongly quasi-tame or feral. This substantially generalizes a recent result of Milani\v{c} and Piva\v{c} [WG'19], who classified all hereditary graph classes defined by a finite set of forbidden induced subgraphs on at most $4$ vertices into tame or feral.
{\bf (ii)} We show that every hereditary class that excludes $k$-creatures and additionally excludes all cycles of length at least $c$, for some constant $c$, is tame. This generalizes the result of [Chudnovsky et al., Arxiv 2019] who obtained the same statement for $c = 5$.
{\bf (iii)} We show that every hereditary class that excludes $k$-creatures and additionally excludes a complete graph on $c$ vertices for some fixed constant $c$ is tame.
{\bf (iv)} Finally we show that the domination requirement in our main result can not be dropped. Specifically we give an example of a feral family that excludes 100-creatures, disproving the aforementioned conjecture of Abrishami et al.

\end{abstract}

\thispagestyle{empty}

\newpage

\setcounter{page}{1}

\section{Introduction}\label{sec:intro}
Let $G$ be a graph and $u$ and $v$ be distinct vertices in $G$. A vertex set $S$ is a $u$-$v$-{\em separator} if $u$ and $v$ are in distinct components of $G-S$. The set $S$ is a {\em minimal} $u$-$v$-{\em separator} if $S$ is a $u$-$v$-separator, but no proper subset of $S$ is a $u$-$v$-separator. Finally, $S$ is a {\em minimal separator} if $S$ is a minimal $u$-$v$-separator for some pair of vertices $u$ and $v$.
Minimal separators have a tremendous role in the design of graph algorithms, both directly, such as in the structural characterization of chordal graphs~\cite{brandstadt1999graph} but also indirectly in optimization algorithms for graph separation and routing problems (for example~\cite{menger1927allgemeinen,Marx06,RobertsonS95b}). The theory of potential maximal cliques, developed by Bouchitt{\'{e}} and Todinca~\cite{BouchitteT01} implies that a several fundamental graph problems, such as computing the {\em treewidth} and {\em minimum fill in} of a graph $G$ can be done in time polynomial in the number of vertices of $G$ and the number of minimal separators in $G$.
Lokshtanov~\cite{Lokshtanov10} showed that the same result holds for computing the {\em tree-length} of the graph $G$, while Fomin et al.~\cite{FominTV15} proved a general result that showed that a whole class of problems (including e.g. {\em maximum independent set} and {\em minimum feedback vertex set}) can be solved in time polynomial in the number of vertices and minimal separators of the graph. All of these algorithms require a list of all the minimal separators of $G$ to be provided as input. However, the listing algorithms for minimal separators of Kloks and Kratsch~\cite{KloksK94} or Berry et al.~\cite{BerryBC00} can be used to compute such a list in time polynomial in the number of vertices times a factor linear in the number of minimal separators of $G$.

This brings to the forefront the main question asked in this paper - {\em which classes of graphs have polynomially many minimal separators?} We will say that a graph class ${\cal F}$ is {\em tame} if there exists an integer $k$ so that every graph in ${\cal F}$ on $n$ vertices has at most $O(n^k)$ minimal separators. %\todo{discuss the terminology/timeline somewhere} 
A number of important graph classes have been shown to be tame, such Chordal~\cite{brandstadt1999graph} (and more generally Weakly Chordal~\cite{BouchitteT01}), Permutation (and, more generally $d$-Trapezoid~\cite{kratsch1996structure}), Circular Arc~\cite{KloksKW98} and Polygon Circle graphs~\cite{suchan2003minimal}. Most of these results date back to the late 1990s and early 2000s.
Much more recently~\cite{abrishami2020graphs,Chudnovsky2019MaximumIS,chudnovsky2019maximum,milani2019minimal}, research has started to focus on a more systematic classification of which graph classes are tame and which are not. Indeed the term {\em tame} for graph classes with polynomially many minimal separators was defined by Milani\v{c} and Piva\v{c}~\cite{milani2019minimal}, who classified all hereditary (closed under vertex deletion) classes defined by a set of forbidden induced subgraphs, all of which have at most $4$ vertices, as tame or not tame. 

Building on the terminology of Milani\v{c} and Piva\v{c}~\cite{milani2019minimal}, we will say that a class of graphs ${\cal F}$ is {\em quasi-tame} if there exist constants $k, c$ such that every $n$-vertex graph in the family contains at most $O(n^{k \log^c n})$ minimal separators. Further, ${\cal F}$ is {\em strongly quasi-tame} if it is quasi-tame with $c \leq 1$. On the opposite side of the spectrum, we will say that  ${\cal F}$ is {\em feral} if there exists a constant $c$ such that for every $N \geq 0$ there exists an $n \geq N$ such that ${\cal F}$ contains an $n$-vertex graph with at least $c^n$ minimal separators.

Abrishami et al~\cite{abrishami2020graphs} define a structure, called a $k$-creature, the presence of which appears to control, to a large extent, whether a graph has many or few minimal separators. 
%The following structure, first introduced in \cite{abrishami2020graphs}, will be a central object of study in this paper.  
A $k$-{\em creature} in a graph $G$ is a four-tuple $(A, B$, $X = \{x_1, x_2, \ldots, x_k\}$, $Y = \{y_1, y_2, \ldots, y_k\})$ of mutually disjoint vertex subsets of $V(G)$, satisfying the following conditions (see Figure \ref{k-creature}). 
\begin{enumerate}\setlength\itemsep{-.6mm}
    \item $A$ and $B$ are connected,
    \item $A$ and $Y \cup B$ are {\em anti-complete} (i.e. no vertex in $A$ is adjacent to a vertex in $B \cup Y$) $B$ is anti-complete with $X \cup Y$.
    \item $A$ dominates $X$ (every vertex in $X$ has a neighbor in $A$) and $B$ dominates $Y$, and
    \item $x_iy_j$ is an edge if and only if $i = j$.
\end{enumerate}

%{\em (a)} $A$ and $B$ each induce connected subgraphs of $G$, {\em (b)} $A$ is anti-complete with $B \cup Y$ (i.e. no vertex in $A$ is adjacent to a vertex in $B \cup Y$) and $B$ is anti-complete with $A \cup X$, {\em (c)} For every $i,j \leq k$, $x_iy_i$ is an edge of $G$ if and only if $i=j$, and {\em (d)} 

%\begin{itemize}\setlength\itemsep{-.6mm}
    %\item $A$ and $B$ each induce connected subgraphs of $G$,
    %\item $A$ is anti-complete with $B \cup Y$ (i.e. no vertex in $A$ is adjacent to a vertex in $B \cup Y$) and $B$ is anti-complete with $A \cup X$.
    
    %\item For every $i,j \leq k$, $x_iy_i \in E(G)$ with $1 \leq i \leq k$, $x_iy_i \in E(G)$, $x_i$ has at least one neighbor in $A$ and $x_i$ and is anti-complete with $B$, $y_i$ has at least one neighbor in $B$ and $y_i$ is anti-complete with $A$.
    
    %\item For $i,j$ with $i \neq j$ and $1 \leq i,j \leq k$ $x_iy_j \notin E(G)$.
%\end{itemize}

\begin{figure}
\centerline{\includegraphics[scale=.6]{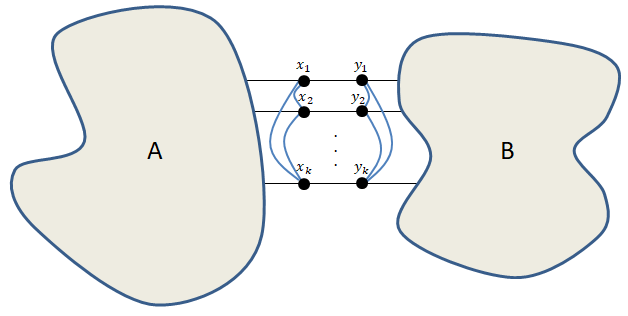}}
\caption{\em A graph induced by the vertices of a k-creature. The blue edges indicate that $x_i$ ($y_i$) may or may not be neighbors with $x_j$ ($y_j$)}
\label{k-creature}
\end{figure}

%has brought attention to the question: {\em which graphs have 

%We denote the edge set of a graph $G$ by $E(G)$ and the vertex set of a graph by $V(G)$. Given a graph $G$, a non-empty set $S \subset V(G)$ is called a $minimal$ $separator$ if there are at least two distinct components $L$ and $R$ of $G-S$ such that $N_G(L) = N_G(R) = S$. A family of graphs ${\cal F}$ is called $tame$ if there exists exists a constant $c$ such that for all $G \in {\cal F}$, $G$ has at most $|V(G)|^c$ minimal separators. A family of graphs ${\cal F}$ is called $strongly$-$quasi$-$tame$ if there exists exists a constant $c$ such that for all $G \in {\cal F}$, $G$ has at most $|V(G)|^{c\log(|V(G)|)}$ minimal separators. A family of graphs ${\cal F}$ is called $feral$ if there exists exists a constant $c > 1$ such that for all natural numbers $N$ there exists a $G \in {\cal F}$, such that $|V(G)| = n > N$ and $G$ has at least $c^n$ minimal separators. A complete discussion of terminology and definitions can be found in Section \ref{sec:prelim}.

%\todo{explain why tamness is a nice property here...}

A graph $G$ is $k$-{\em creature-free} if there does not exists a tuple of vertex sets of $V(G)$ that form a $k$-creature. It is easy to see that a $k$-creature contains at least $2^k$ minimal separators (select precisely one of $\{x_i, y_i\}$ for every $i \leq k$). Because deleting a vertex can not increase the number of minimal separators, a graph $G$ that contains a $k$-creature contains at least $2^k$ minimal separators. Thus, a graph family ${\cal F}$ that contains $n$-vertex graphs with $k$-creatures for arbitrarily large $n$ and with $k = \Omega(k)$ is feral. For ${\cal F}$ to not be tame it is sufficient for $k$ to grow super-logarithmically with $n$ (i.e $n = 2^{o(\log k)}$). A sort of converse to this observation was conjectured in \cite{abrishami2020graphs}. 
%More specifically, the following was conjectured in \cite{abrishami2020graphs}.

\begin{conjecture}\cite{abrishami2020graphs}\label{false conjecture} 
%%%%%%%%%%%%%%%%%%%There exists a function $f : \mathbb{N} \rightarrow \mathbb{N}$ such that if $G$ is $k$-creature free, then $G$ has at most $|V(G)|^{f(k)}$ minimal separators. 
For every fixed natural number $k$, the family of graphs that are $k$-creature-free is tame.
\end{conjecture}

%
%It will simplify some notation if we assume that the graphs we consider always contain a 1 creature. Every graph with at least one minimal separator contains a 1-creature, so this assumption is without loss of generality.
%

%, and therefore any graph family that contains graphs with $k$-creatures for arbitrarily large $k$ and the number of vertices of these graphs grows linearly with $k$ are feral. For the family to not be tame it is sufficient for $k$ to be super-logarithmic in $n$.

%This happens for instance in hereditary graph families that are also closed under edge contraction and have graphs with $k$-creatures for arbitrarily large $k$. 

Even if Conjecture~\ref{false conjecture} were to be true, it would still not give a complete characterization of hereditary graph classes into tame or non-tame. In particular Abrishami et al~\cite{abrishami2020graphs} give an example of a tame hereditary class ${\cal F}$ that contains arbitrarily large $k$-creatures. Their example can also be slightly modified to show that there exist hereditary families that are neither tame nor feral. This makes it appear that, at least for hereditary classes in their full generality, the boundary between tame and non-tame graph classes is so ``strange-looking'' that a complete dichotomy may be out of reach, and that we therefore have to settle for sufficient conditions for tameness / non-tameness, and possibly complete characterizations for more well-behaved sub-classes of hereditary families. 
For an example, Conjecture~\ref{false conjecture}, if true, would have yielded a complete dichotomy into tame or feral for all classes of graphs closed under {\em induced minors} (i.e closed under vertex deletion and edge contraction). 

Unfortunately it turns out that {\bf Conjecture~\ref{false conjecture} is false}. In particular we give (in Section~\ref{counter example section}) an example of a feral family ${\cal F}$ that excludes $100$-creatures. The family ${\cal F}$ consists of all $k$-{\em twisted ladders} (see Section~\ref{counter example section}). Our main result is nevertheless that Conjecture~\ref{false conjecture} is true ``in spirit'', in the sense that for large classes of hereditary families, excluding $k$-creatures does imply few minimal separators. To state Theorem~\ref{theorem 1} we need to define $k$-{\em skinny ladders}. A $k$-{\em skinny-ladder} is a graph $G$ consisting of two paths $P_l = \ell_1\ell_2 \ldots \ell_k$ and $P_r = r_1r_2 \ldots r_k$ and a set $\{s_1, s_2,\ldots, s_k\}$ of vertices such that for every $i$, $s_i$ is adjacent to $\ell_i$ and $r_i$ and to no other vertices. 

%We will give a counterexample to this conjecture in Section \ref{counter example section} by constructing a family of graphs where each graph is $k$-creature free for a fixed $k$ chosen large enough, but the family of graphs is feral. We will also prove a modification of Conjecture \ref{false conjecture}, stated in Theorem \ref{theorem 1}, although we are only able to prove a quasi-polynomial bound in Theorem \ref{theorem 1}. 

%We call a graph $G$ a $k$-$skinny$-$ladder$ if the following two conditions hold:

%\begin{itemize}
    %\item $V(G) = \{\ell_1, \ell_2, \ldots, \ell_k,\} \cup \{s_1, s_2,\ldots, s_k\} \cup \{r_1, r_2,\ldots,r_k\}$
    
    %\item For all $i$ with $1 \leq i < k$ $\ell_i\ell_{i+1} \in E(G)$ and $r_ir_{i+1} \in E(G)$, for all $i$ with $1 \leq i \leq k$ $\ell_is_i \in E(G)$ and $s_ir_i \in E(G)$ and $G$ has no other edges.
%\end{itemize}

\begin{theorem}\label{theorem 1}
%%%%%%%%%%%%%%%%%%%There exists a function $f : \mathbb{N} \rightarrow \mathbb{N}$ such that if $G$ is $k$-creature free and $G$ does not contain a $k$-skinny-ladder as an induced minor, then $G$ has $|V(G)|^{f(k)\log(|V(G)|)}$ minimal separators.
For every natural number $k$, the family of graphs that are $k$-creature free and do not contain a $k$-skinny-ladder as an induced minor is strongly-quasi-tame.
\end{theorem}

Theorem \ref{theorem 1} suggests that other counterexamples to Conjecture \ref{false conjecture} should resemble the counterexample we provide in Section \ref{counter example section}. We do not have any examples of classes that are strongly quasi-tame according to Theorem~\ref{theorem 1}, and conjecture that the statement of Theorem~\ref{theorem 1} remains true even if  strongly quasi-tame is replaced by tame.  

Excluding the $k$-skinny ladder is closely tied to {\em domination} of minimal separators. A vertex set $X$ {\em dominates} $S$ of every vertex in $S$ is either in $X$ or has a neighbor in $X$.
An important ingredient in the proof of Theorem~\ref{theorem 1} (see Lemma~\ref{domination or creature/ladder}) is that for every $k$ there exists a $k'$ so that if $G$ excludes $k$-creatures and $k$-skinny ladders as induced minors then every minimal separator $S$ in $G$ is dominated by a set $X$ on at most $k'$ vertices. 
In fact, because a $k$-skinny ladder is itself $5$-creature-free and contains a minimal separator (namely $S$) which can not be dominated by $k-1$ vertices, among the hereditary classes ${\cal F}$ that exclude $k$-creatures, the presence or absence of $k$-skinny ladders precisely characterizes whether every minimal separator of every graph in ${\cal F}$ can be dominated by a constant size set of vertices.

To demonstrate the power of Theorem~\ref{theorem 1} we show that it gives, as a pretty direct consequence, a complete classification of all hereditary graph classes defined by a finite set of forbidden induced subgraphs into strongly quasi-tame or feral. Indeed, it is an easy exercise to show that if a family ${\cal F}$ is defined by a finite set of forbidden induced subgraphs and contains skinny $k$-ladders for arbitrarily large $k$, then there exists a constant $f$ such that ${\cal F}$ either contains all $f$-subdivisions of $3$-regular graphs (an $f$-{\em subdivision} of $G$ is the graph obtained from $G$ by replacing each edge of $G$ by a path on $f+1$ edges) or all line graphs~(see \cite{DiestelBook} for a definition) of $f$-subdivisions of $3$-regular graphs. In this case ${\cal F}$ is feral. Therefore, Theorem~\ref{theorem 1} proves Conjecture~\ref{false conjecture} for hereditary graph classes defined by a finite set of forbidden induced subgraphs, albeit with strongly quasi-tame instead of tame. We obtain Theorem~\ref{theorem 2} by extracting a small set of graphs that themselves contain large $k$-creatures, such that graphs that contain $k$-creatures contain one of them as an induced subgraph. We refer to Figure~\ref{forbidden graphs} as well as to Section~\ref{sec:prelim} for the definitions of the graphs used in Theorem~\ref{theorem 2}

\begin{theorem}\label{theorem 2}
%%%%%%%%%%%%%%%%%%Let ${\cal F}$ be a graph family defined by a finite number of forbidden induced subgraphs. There exists a function $f : \mathbb{N} \rightarrow \mathbb{N}$ such that if there exists $k$ so that ${\cal F}$ forbids all $k$-thetas, $k$-prisms, $k$-pyramids, $k$-ladder-thetas, $k$-ladder-prisms, $k$-claws, and $k$-paws, then for all $G \in {\cal F}$ the number of minimal separators of $G$ is at most $n^{f(k)\log(n)}$. Otherwise there exists a constant $c > 1$ such that for arbitrarily large $N$ there exists graphs $G \in {\cal F}$ such that $|V(G)| = n > N$ and $G$ has at least $c^n$ minimal separators.

Let ${\cal F}$ be a graph family defined by a finite number of forbidden induced subgraphs. If there exists a natural number $k$ such that ${\cal F}$ forbids all $k$-theta, $k$-prism, $k$-pyramid, $k$-ladder-theta, $k$-ladder-prism, $k$-claw, and $k$-paw graphs, then ${\cal F}$ is strongly-quasi-tame. Otherwise ${\cal F}$ is feral.
\end{theorem}

Note that some of the graphs of Figure \ref{forbidden graphs} share a name with graphs that appear in the work of Abrishami et al.~\cite{abrishami2020graphs}, but the definitions given here are slightly different. In particular, in some of the places where they require single edges we allow arbitrarily long paths. Abrishami et al.~\cite{abrishami2020graphs} prove that the family of (what they define to be) theta-free, pyramid-free, prism-free, and turtle-free graphs is tame. We remark that our results are incomparable to theirs, in the sense that there are classes of graphs whose tameness follows from their work, but not ours, and vice versa. 

%
%
%%%%%% MAYBE MOVE AND INCLUDE SOMEWHERE ELSE
%
%the following example (we have slightly modified it to better fit our notation): Let $T_k$ be the graph consisting of two paths $P_A$, and $P_B$ on $2^k \cdot k$ vertices each, together with sets $X = \{x_1, \ldots, x_k\}$ and $Y = \{y_1, \ldots, y_k\}$, where $x_i$ is adjacent to $y_i$ and all vertices on $P_A$ numbered $2^k \cdot (i-1) + 1$ through $2^k \cdot i$, and $y_i$ is adjacent to $x_i$ and all vertices on $P_B$ numbered $2^k \cdot (i-1) + 1$ through $2^k \cdot i$. Let ${\cal F}$ be the class of all induced subgraphs of $T_k$ for all $k$. This class clearly contains $k$-creatures for arbitrarily large $k$, however the number of minimal separators in $T_k$ is $O(2^k)$, which is linear in the number of vertices. Some additional work (omitted) shows that the number of minimal separators in linear in the number of vertices for all induced subgraphs of $T_k$ as well. 
%
%

%In Theorem \ref{theorem 2} we provide a complete characterization for when a family of graphs defined by a finite number of forbidden subgraphs is strongly-quasi-tame (See Section \ref{sec:prelim} for the definition of a graph family defined by a finite number of forbidden induced subgraphs and related definitions). 

Theorem~\ref{theorem 2} substantially generalizes the main result of Milani\v{c} and Piva\v{c}~\cite{milani2019minimal}, who obtained a complete classification into tame or feral of hereditary graph classes characterized by forbidden induced subgraphs on at most $4$ vertices. The generalization comes at a price - as our upper bounds on the number of minimal separators are quasi-polynomial instead of polynomial.

Finally, we explore for which classes we are able to improve our quasi-polynomial upper bounds to polynomial ones. Here, again, domination plays a crucial role. We show that for every pair $k$, $k'$ of integers every class of graphs that excludes $k$-creatures and additionally has the property that every minimal separator $S$ is dominated by a set $X$ of size $k'$ {\em disjoint from} $S$ is tame. We then proceed to show that graphs that exclude $k$-creatures and all cycles of length at least $t$ for some positive integer $t$, leading to the following result.

%at the cost of 
%quasi-polynomial overhead. In \cite{milani2019minimal} the authors characterize when a family of graphs defined by a finite number of forbidden induced subgraphs, where every explicitly forbidden induced subgraph has five or less vertices, is tame. 

\begin{figure}
\centerline{\includegraphics[width=0.95\textwidth]{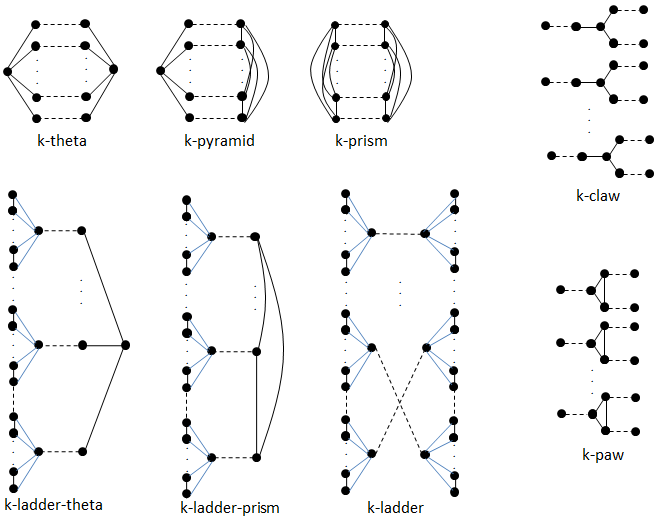}}
\caption{\em Note that for $k$-ladders we only depict one of the many possibilities of how the ordering of the neighborhoods of the $a_i$'s and $b_i$'s may be arranged in $L$ and $R$. Dashed lines represent the option of having an arbitrary length path or just an edge (except for $k$-claw and $k$-paw graphs which the dotted line is always a path of length $k$.) The blue lines used in the $k$-ladder-theta, $k$-ladder-prism, and $k$-ladder graphs represents the option of either having or not having that edge, but for each vertex adjacent to more than one of the blue edges, at least one of those blue edges must belong to the graph. 
}
\label{forbidden graphs}
\end{figure}

%%%%%%%%%%%%%%%%%%%Let ${\cal F}$ be a graph family defined by a finite number of forbidden induced subgraphs. We can then see that Theorem \ref{theorem 2} provides a necessary and sufficient condition for when there exists a quasi-polynomial function $p$ such that for all $G \in {\cal F}$, $G$ has at most $p(n)$ minimal separators.

%%%%%%%%%%%%%%%%%%%A graph $G$ is called $C_{\geq r}$-free if it contains no induced cycles of length $r$ or more, where the length of a cycle is the number of vertices it contains. In \cite{chudnovsky2019maximum} the authors prove as their main theorem that there exists a function $f : \mathbb{N} \rightarrow \mathbb{N}$ such that if $G$ is a $C_{\geq 5}$-free graph and is $k$-prism free, then $G$ has at most $|V(G)|^{f(k)}$ minimal separators. Theorem \ref{theorem 3} generalizes this result.

\begin{theorem}\label{theorem 3}
%%%%%%%%%%%%%%%%%%%There exists a function $f : \mathbb{N} \rightarrow \mathbb{N}$ such that if $G$ is a $C_{\geq k}$-free graph and $G$ is $k$-theta, $k$-prism, and $k$-pyramid free for some natural number $k$, then $G$ has $|V(G)|^{f(k)}$ minimal separators.
For every pair of natural numbers $k$ and $r$, the family of graphs that are $C_{\geq k}$-free, $k$-theta-free, $k$-prism-free, and $k$-pyramid-free is tame.
\end{theorem}

Here a graph $G$ is $C_{\geq r}$-{\em free} if it contains no induced cycles of length at least $r$. Theorem~\ref{theorem 3} is optimal in the sense that $k$-theta, $k$-prism, and $k$-pyramid graphs have at least $2^{k-2}$ minimal separators and therefore can have exponentially many minimal separators. Further, it substantially strengthens the results of Chudnovsky et al.~\cite{chudnovsky2019maximum}, who prove the same statement but only for $r = 5$. 

Finally we show that also graph classes that exclude $k$-creatures, $k$-skinny ladders, as well as at least one clique satisfy the property that every minimal separator $S$ can be dominated by a constant size set $X$ disjoint from $S$. This implies that also this family of graphs is tame. 

%as complete graphs of size at least $t$ for some positive integer $t$, have this property.

%prove that the family of graphs that are $C_{\geq 5}$-free (also known as long-hole free) and $k$-prism-free is tame. Theorem \ref{theorem 3} generalizes this result. 

%Lastly, we will show that if we restrict the (minimal separators of the) graphs of our first two theorems to have bounded clique size, then we can improve the quasi-polynomial bounds of Theorems \ref{theorem 1} and \ref{theorem 2} to polynomial bounds. In particular we will prove the following two theorems.

\begin{theorem}\label{theorem 4}
%%%%%%%%%%%%%%%%%%%%%%%%%There exists a function $f : \mathbb{N} \rightarrow \mathbb{N}$ such that if $G$ is $k$-creature free, $G$ does not contain a $k$-skinny-ladder as an induced minor, and no minimal separator of $G$ has at most a clique of size $k$, then $G$ has $|V(G)|^{f(k)}$ minimal separators.

For any fixed natural number $k$, the family of graphs that are $k$-creature-free, contain no $k$-skinny-ladder as an induced minor, and contain no minimal separator that has a clique of size $k$ is tame.
\end{theorem}

\begin{theorem}\label{theorem 5}
%%%%%%%%%%%%%%%%%%%%%%%%%%Let ${\cal F}$ be a graph family defined by a finite number of forbidden induced subgraphs. There exists a function $f : \mathbb{N} \rightarrow \mathbb{N}$ such that if ${\cal F}$ forbids all $k$-cliques, $k$-thetas,  $k$-ladder-thetas, $k$-claws, and $k$-paws, then all $G \in {\cal F}$ $G$ has at most $n^{f(k)}$ minimal separators. Otherwise, ${\cal F}$ contains all cliques or there exists a constant $c > 1$ such that for arbitrarily large $N$ there exists graphs $G \in {\cal F}$ such that $|V(G)| = n > N$ and $G$ has at least $c^n$ minimal separators. 

Let ${\cal F}$ be a graph family defined by a finite number of forbidden induced subgraphs. If there exists a natural number $k$ such that ${\cal F}$ forbids all $k$-clique, $k$-theta, $k$-ladder-theta, $k$-claw, and $k$-paw graphs then ${\cal F}$ is tame. Otherwise, ${\cal F}$ contains all cliques or ${\cal F}$ is feral.
\end{theorem}

\section{Preliminaries}\label{sec:prelim}
All graphs in this paper are assumed to be simple, undirected graphs. We denote the edge set of a graph $G$ by $E(G)$ and the vertex set of a graph by $V(G)$. If $v$ $\in$ $V(G)$, then we use $N_G[v]$ to denote the closed neighborhood of $v$ in the graph $G$, i.e. the set of all neighbors $v$ has in $G$ together with $v$ itself. We use $N_G(v)$ to denote the set $N_G[v] - \{v$\}. If $X$ $\subseteq$ $V(G)$, then $N_G[X]$ = $\bigcup_{x \in X}N_G[x]$ and $N_G(X)$ = $N_G[X] - X$. When the graph $G$ is clear from the context, we will use $N[v]$, $N(v)$, $N[X]$, and $N(X)$. If $X \subset V(G)$, then we use $G[X]$ to denote the induced subgraph of $G$ with vertex set $X$ and $G-X$ denotes $G[V(G)-X]$.

Given a graph $G$, a non-empty set $S \subset G$ is called a minimal separator if there are at least two distinct connected components $L$ and $R$ of $G-S$ such that $N_G(L) = N_G(R) = S$. If $u \in L$ and $v \in R$ then we call $S$ a $minimal$ $u$-$v$-$separator$ or a  $u,v$-$minimal$ $separator$. We say a component $X$ of $G-S$ is an $S$-full component if $N_G(X) = S$. A family of graphs ${\cal F}$ is called $tame$ if there exists exists a constant $c$ such that for all $G \in {\cal F}$, $G$ has at most $|V(G)|^c$ minimal separators. A family of graphs ${\cal F}$ is called $strongly$-$quasi$-$tame$ if there exists exists a constant $c$ such that for all $G \in {\cal F}$, $G$ has at most $|V(G)|^{c\log(|V(G)|)}$ minimal separators. A family of graphs ${\cal F}$ is called $feral$ if there exists exists a constant $c > 1$ such that for all natural numbers $N$ there exists a $G \in {\cal F}$, such that $|V(G)| = n > N$ and $G$ has at least $c^n$ minimal separators.

Given a path $P = v_1,v_2,...,v_k$ we call $v_1$ and $v_k$ the $endpoints$ of $P$, and all other vertices of $P$ are $internal$ vertices of $P$. The $length$ of a path is the number of vertices in the path.
Given a graph $G$ and a graph $H$, $G$ is said to be $H$-free or $G$ forbids $H$ if $G$ does not contains $H$ as an induced subgraph. We will sometimes talk about the induced minors of a graph so being $H$-free should not be confused with $G$ not containing $H$ as an induced minor. If $G$ does not contain $H$ as an induced minor then that is precisely what we will say, that $G$ does not contain $H$ as an induced minor. 
If ${\cal G}$ is a family of graphs such that every $G \in {\cal G}$ is $H$-free, then ${\cal G}$ is said to be $H$-free or that ${\cal G}$ forbids $H$. Similarly, given a graph $G$ and a family of graph ${\cal H}$, $G$ is said to be ${\cal H}$-free or $G$ forbids ${\cal H}$ if $G$ does not contain any $H \in {\cal H}$ as an induced subgraph. If ${\cal G}$ is a family of graphs such that every $G \in {\cal G}$ is ${\cal H}$-free, then ${\cal G}$ is said to be ${\cal H}$-free or that ${\cal G}$ forbids ${\cal H}$.

Let ${\cal F}$ be a family of graphs. We say that ${\cal F}$ is a $family$ $of$ $graphs$ $defined$ $by$ $a$ $finite$ $number$ $of$ $forbidden$ $induced$ $subgraphs$ if there exists a finite set of graphs ${\cal H}$ such that $G \in {\cal F}$ if and only if $G$ does not contain an induced subgraph isomorphic to any graph in ${\cal H}$. We say that ${\cal H}$ is a set of forbidden subgraphs that define ${\cal F}$, and if $H \in {\cal H}$ we say that ${\cal F}$ explicitly forbids $H$. 
%If no graph in ${\cal F}$ contains a graph $K$ as an induced subgraph, we say that ${\cal F}$ forbids $K$.

Given a graph $G$ let $H$ and $K$ be two subsets of $V(G)$. We say that $H$ is $anti$-$complete$ with $K$ or that $H$ and $K$ are $anti$-$complete$ if $H$ and $K$ are disjoint and every vertex in $H$ is non-adjacent to every vertex in $K$ in $G$. We extend this definition in an obvious way to allow $H$ (and possibly $K$) to be a subgraph of $G$ by saying $H$ is anti-complete with $K$ if $V(H)$ is anti-complete with $K$ ($V(K)$ if $K$ is also a subgraph).  A set $X \subset V(G)$ is said to dominate a set $Y \subset V(G)$ if for every $y \in Y$ either $y \in X$ or there is an $x \in X$ such that $yx \in E(G)$.
%If $H, K \subset V(G)$ are two disjoint vertex sets, then we say that $H$ is $anti$-$complete$ with $K$ or that $H$ and $K$ are $anti$-$complete$ if every vertex in $G[H]$ is non-adjacent to every vertex in $G[K]$.

Given a graph $G$ and an edge $uv \in E(G)$ we denote by $G^{uv}$ the graph that results from $contracting$ the edge $uv$ in $G$, so $V(G^{uv}) = (V(G)-\{u,v\})\cup\{w\}$ and for $x,y \in V(G)-\{u,v\}$, $xy \in E(G^{uv})$ if and only if $xy \in E(G)$ and for $x \in V(G)-\{u,v\}$, $xw \in E(G^{uv})$ if and only if $x$ is neighbors with $u$ and/or $v$ in $G$. Given an induced path $P$ of $G$, we denote by $G^P$ the graph that results from contracting each edge of $P$ one at a time. Note that the resulting graph is independent of the order the edges are contracted in.

Given two anti-complete graphs $A$ and $B$ with $a\in A$ and $b \in B$, we define an operation $gluing$ $a$ $to$ $b$ which produces a new graph $C$ where $V(C) = [(V(A) \cup V(B))-\{u,v\})] \cup \{w\}$ and for $x,y \in V(C)-\{w\}$, $xy \in E(C)$ if and only if $xy \in E(A)$, or $xy \in E(B)$ and for $x \in V(C)-\{w\}$, $xw \in E(G^{uv})$ if and only if $x$ is neighbors with $a$ in $A$ and/or $b$ in $B$. This is the exact same graph that would result in adding an edge between $a$ and $b$ and then contracting that edge.

\subsection{Graph Definitions}

Given a graph $G$ we call a tuple $(A, B$, $\{x_1, x_2, \ldots, x_k\}$, $\{y_1, y_2, \ldots, y_k\})$ of mutually disjoint vertex subsets of $V(G)$ a $k$-$creature$ if the following conditions hold: (see Figure \ref{k-creature} for a depiction of the graph a $k$-creature induces)

\begin{itemize}
    \item $G[A]$ and $G[B]$ are connected and $A$ is anti-complete with $B$.
    
    \item for $i$ with $1 \leq i \leq k$, $x_iy_i \in E(G)$, $x_i$ has at least one neighbor in $A$ and $x_i$ and is anti-complete with $B$, $y_i$ has at least one neighbor in $B$ and $y_i$ is anti-complete with $A$.
    
    \item for $i,j$ with $i \neq j$ and $1 \leq i,j \leq k$ $x_iy_j \notin E(G)$.
\end{itemize}

A graph is said to be $k$-creature free if there does not exists a tuple of vertex sets of $V(G)$ that form a $k$-creature. 
%It will simplify some notation if we assume that the graphs we consider always contain a 1 creature. Every graph with at least one minimal separator contains a 1-creature, so this assumption is without loss of generality.

We call a graph $G$ a $k$-$skinny$-$ladder$ if the following conditions hold:

\begin{itemize}
    \item $V(G) = \{\ell_1, \ell_2, \ldots, \ell_k,\} \cup \{s_1, s_2,\ldots, s_k\} \cup \{r_1, r_2,\ldots,r_k\}$
    
    \item For all $i$ with $1 \leq i < k$ $\ell_i\ell_{i+1} \in E(G)$ and $r_ir_{i+1} \in E(G)$, for all $i$ with $1 \leq i \leq k$ $\ell_is_i \in E(G)$ and $s_ir_i \in E(G)$, and $G$ has no other edges.
\end{itemize}

We call a graph $G$ a $k$-$almost$-$skinny$-$ladder$ if the following conditions hold:

\begin{itemize}
    \item $V(G) = L \cup S \cup R$ with $L$, $S$, and $R$ mutually disjoint and $|S| = k$.
    
    \item $G[L]$ and $G[R]$ form induced paths of $G$ and $L$ is anti-complete with $R$.
    
    \item Each $s \in S$ has at least one neighbor in $L$ and at least one neighbor in $R$.
    
    \item For all pairs $x, y \in S$, if $a,b$ are neighbors of $x$ in $L$, then $y$ has no neighbors on the subpath of $G[L]$ that has $a$ and $b$ as its endpoints. Similarly, if $a,b$ are neighbors of $x$ in $R$, then $y$ has no neighbors on the subpath of $G[R]$ that has $a$ and $b$ as its endpoints. This last condition requires that no vertex of $L$ or $R$ has more than one neighbor in $S$.
    
\end{itemize}

The following graphs, except for $k$-ladder graphs, appear in Theorem \ref{theorem 2}. Figure \ref{forbidden graphs} depicts these graphs. It can be seen that all graphs here except for $k$-claw and $k$-paw graphs contains at least $2^{k-2}$ minimal separators. 

\begin{itemize}
    \item A graph $G$ is a $k$-$theta$ if $G$ consist of two vertices $a, b$ and $k$ induced paths $P_1, P_2, \ldots P_k$. For $1 \leq i \leq k$ the end points of $P_i$ are $a$ and $b$, every $P_i$ is anti-complete with $P_j$, and every $P_i$ has length at least 4.
    
    \item A graph $G$ is a $k$-$prism$ if $G$ consist of two disjoint cliques $a_1, a_2, \ldots, a_k$ and $b_1, b_2, \ldots, b_k$ along with $k$ induced paths $P_1, P_2, \ldots P_k$ each of length at least 2. For $1 \leq i \leq k$ the end points of $P_i$ are $a_i$ and $b_i$, every $P_i-\{a_i, b_i\}$ is anti complete with $P_j-\{a_j, b_j\}$, and $a_i$ is neighbors with $b_j$ if and only if $i = j$ and $P_i$ is a path of length 2.

    \item A graph $G$ is a $k$-$pyramid$ if $G$ consist of a vertex $a$ and a clique $b_1, b_2, \ldots, b_k$, where $a$ is anti-complete with $b_1, b_2, \ldots, b_k$, along with $k$ induced paths $P_1, P_2, \ldots P_k$ each of length at least 3. For $1 \leq i \leq k$ the end points of $P_i$ are $a$ and $b_i$ and every $P_i-\{a, b_i\}$ is anti complete with $P_j-\{a, b_j\}$.
    
    \item A graph $G$ is a $k$-$ladder$-$theta$ if $G$ consists of an induced path $L$ and a vertex $b$ anti-complete with $L$, along with $k$ mutually disjoint induced paths $P_1, P_2, \ldots P_k$ that are also disjoint from $L$, each of length at least 3. For $1 \leq i \leq k$ the end points of $P_i$ are $a_i$ and $b$, every $P_i-\{b\}$ is anti-complete with $P_j-\{b\}$, $P_i -\{a_i\}$ is anti-complete with $L$, every $a_i$ has at least one neighbor in $L$, and if $x,y$ are neighbors with $a_i$ in $L$, then no $a_j$ with $i \neq j$ has a neighbor in the induced subpath of $L$ that has $x$ and $y$ as its endpoints.
    
    \item A graph $G$ is a $k$-$ladder$-$prism$ if $G$ consists of an induced path $L$ and clique $b_1, b_2, \ldots, b_k$ where $L$ is anti-complete with $b_1, b_2, \ldots, b_k$, along with $k$ induced paths $P_1, P_2, \ldots P_k$ each of length at least 2. For $1 \leq i \leq k$ the end points of $P_i$ are $a_i$ and $b_i$, every $P_i-\{b_i\}$ is anti-complete with $P_j-\{b\}$, $P_i -\{a_i\}$ is anti-complete with $L$, every $a_i$ has at least one neighbor in $L$, and if $x,y$ are neighbors with $a_i$ in $L$, then no $a_j$ with $i \neq j$ has a neighbor in the induced subpath of $L$ that has $x$ and $y$ as its endpoints.

   \item A graph $G$ is a $k$-$ladder$ if $G$ consists of two disjoint paths $L$ and $R$, along with $k$ disjoint induced paths $P_1, P_2, \ldots P_k$ each of length at least 2. For $1 \leq i \leq k$ the end points of $P_i$ are $a_i$ and $b_i$, every $P_i-\{b_i\}$ is anti-complete with $P_j-\{b_i\}$, $P_i -\{a_i\}$ is anti-complete with $P_j-\{a_j\}$, every $a_i$ has at least one neighbor in $L$, every $b_i$ has at least one neighbor in $R$, if $x,y$ are neighbors with $a_i$ in $L$, then no $a_j$ with $i \neq j$ has a neighbor in the induced subpath of $L$ that has $x$ and $y$ as its endpoints, and if $x,y$ are neighbors with $b_i$ in $R$, then no $b_j$ with $i \neq j$ has a neighbor in the induced subpath of $R$ that has $x$ and $y$ as its endpoints. Note that we do not have any requirements on the ordering that the neighborhoods the $a_i$'s and $b_i$'s have into $L$ and $R$ respectively (so it could happen that $a_i$'s neighborhood in $L$ may lie in between $a_j$'s and $a_{\ell}$'s neighborhood in $L$, but $b_i$'s neighborhood in $R$ does not lie in between $b_j$'s and $b_{\ell}$'s neighborhood in $R$, this is illustrated in the $k$-ladder given in Figure \ref{forbidden graphs}. We could force this not to happen though with an easy application of the Erdos-Szekers Lemma at the cost of making it a $\sqrt{k}$-ladder.) 
    
    \item A graph $G$ is a $k$-$claw$ if $G$ consists of $k$ disjoint, anti-complete copies of the following graph which we call a $long$-$claw$ $of$ $arm$ $length$ $k$: let $v$ be a vertex and $P_1$, $P_2$, $P_3$ be three paths of length $k$ each with $v$ as one of its endpoints and $P_i-\{v\}$ is anti-complete and disjoint with $P_j-\{v\}$ (i.e. the graph is a claw with each edge subdivide $k-2$ times)
    
    \item A graph $G$ is a $k$-$paw$ if $G$ consists of $k$ disjoint, anti-complete copies of the following graph which we call a $long$-$paw$ $of$ $arm$ $length$ $k$: let $v_1, v_2, v_3$ be a triangle and $P_1$, $P_2$, $P_3$ be three disjoint induced paths of length $k$ each such that $P_i$ has $v_i$ as one of its endpoints and $P_i-\{v_i\}$ is anti-complete with $P_j-\{v_j\}$ for $1 \leq i \neq j \leq 3$.
    
\end{itemize}

%\newpage

%\section{General Purpose Lemmas}\label{sec:p5alg}
%\input{General Lemmas.tex}

\section{A $k$-Creature-Free Feral Graph Family}\label{counter example section}
%\begin{wrapfigure}{c}{0.5\textwidth} 
%\begin{figure}
%\centerline{\includegraphics{twisted ladder 4.png}}
%\caption{The $k$-twisted-ladder}
%\label{twisted ladder}
%\end{figure}
%\end{wrapfigure}

\begin{wrapfigure}{c}{0.4\textwidth}
\begin{center} \hspace{1cm}\includegraphics[scale=.8]{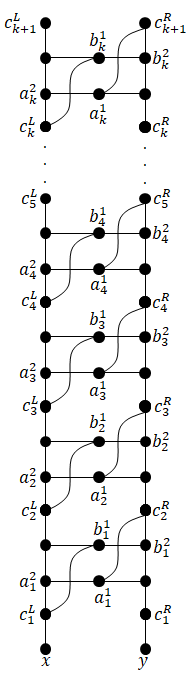}
%\vspace{-1cm}
\caption{\em The $k$-twisted-ladder.}\label{twisted ladder}
\end{center} %\vspace{-8mm} 
\end{wrapfigure}

In this section we will show that the graph of Figure \ref{twisted ladder}, which we will refer to as the $k$-twisted ladder, is a counterexample to Conjecture \ref{false conjecture}. We begin next paragraph by giving a few definitions, then in the following paragraph we will observe that the $k$-twisted-ladder has $2^k$ minimal separators, and finally Lemma \ref{counter example} completes the counterexample by showing that the $k$-twisted-ladder does not contain a large $k$-creature.

We defined a partition of the vertices as follows, let $S$ denote the set of labeled vertices of the $k$-twisted ladder that have 1 as their superscript. If we remove $S$ from the $k$-twisted-ladder we get two induced path, one on the left side which we will refer to as $L$ and one on the right side which we will refer to as $R$. 
%We will also number the vertices of the $k$-twisted ladder as follows. Number the vertices of $L$ $1$ through $|V(L)|$ where vertex $x$ gets numbered as $1$ and vertex numbered $i$ is neighbors with vertex numbered $i-1$ and $i+1$. Similarly number the vertices of $R$ $1$ through $|V(R)|$ where vertex $y$ gets numbered $1$ and vertex numbered $i$ is neighbors with vertex numbered $i-1$ and $i+1$. Number the vertices of $S$ based on the highest numbered neighbor in $L$ that it has. For the vertex $v$ in a $k$-twisted-ladder, we let $n(v)$ denote the number this vertex has been given.
We also define the $i^{th}$ $block$ of the $k$-twisted ladder to be the set of vertices that contains the vertices of the subpath of $L$ that has $c_{i+1}^L$ and $c_{i}^L$ as its endpoints, the vertices of the subpath of $R$ that has $c_{i+1}^R$ and $c_{i}^R$ as its endpoints, and the vertices $a_i^1$ and $b_i^1$. So, the $i^{th}$ block and the $(i+1)^{th}$ block overlap at the vertices $c_{i+1}^R$ and $c_{i+1}^L$.

To see that the $k$-twisted ladder has at least $2^k$ minimal separators we make a set $X$. For each $i$ with $1 \leq i \leq k$ we choose $j \in \{1,2\}$ and add $a^j_i$ and  $b^j_i$ to $X$. $X$ is then an $x,y$-minimal separator, and there are $2^k$ different choices we had when making $X$, so $k$-twisted ladder has at least $2^k$ minimal separators.

To complete the counterexample, we show in the following lemma that this structure does not have a large $k$-creature. To make the result as easy as possible to verify, we show no $k$-twisted-ladder has a $100$-creature, although a significantly smaller upper bound exists.

\begin{lemma}\label{counter example}
 $k$-twisted-ladders are 100-creature-free for all k.
\end{lemma}

\begin{proof}

Let $H$ be a $k$-twisted ladder. Assume for a contradiction that $H$ contains a $100$-creature $(A, B$, $\{x_1, x_2, \ldots, x_{100}\}$, $\{y_1, y_2, \ldots, y_{100}\})$.
%, and without loss of generality assume that $A$ and $B$ are minimum connected sets ($|V(A)|$ and $|V(B)|$ are as small as possible) such that $\{x_1, x_2, \ldots, x_{100}\} \subset N(A)$ and $\{y_1, y_2, \ldots, y_{100}\} \subset N(B)$.

Let $X_A$ and $X_B$ denote the highest numbered block that $A$ and $B$ have a vertex in respectively, and let $Y_A$ and $Y_B$ denote the lowest numbered block that $A$ and $B$ have a vertex in respectively. Let $i =  max (Y_A, Y_B) + 1$ and let $j = min (X_A, X_B) - 1$. Let $k$ be an integer such that $i \leq k \leq j$ (If no such $k$ exists, then the only blocks that can contain vertices from both $A$ and $B$ must be two adjacent blocks and it is clear the lemma is true in this case). Then since $A$ and $B$ are connected we can see by inspection that $A$ must contain one vertex from $\{c_k^L, c_k^R\}$ and $\{c_{k+1}^L, c_{k+1}^R\}$ and $B$ must contain one vertex from $\{c_k^L, c_k^R\}$ and $\{c_{k+1}^L, c_{k+1}^R\}$. Furthermore, since $A$ is anti-complete with $B$, we can again see from inspection that if $c_k^L \in A$ then we must have $c_{k+1}^L \in A$, $c_k^R \in B$, and $c_{k+1}^R \in B$ (since removal of the closed neighborhood of any path from $c_k^L$ to $c_{k+1}^R$, in fact even just the removal of the closed neighborhoods of $c_k^L$ and $c_{k+1}^R$, would separate blocks numbered greater than $k$ from blocks numbered less than $k$), and if $c_k^R \in A$ then we must have $c_k^L \in B$, and by the same reasoning it follows that $c_{k+1}^L \in B$, and therefore $c_{k+1}^R \in A$. Therefore, without loss of generality we may assume that for all $k$ with $i \leq k \leq j$ that $c_k^L \in A$ and $c_k^R \in B$. It then follows from this assumption and the fact that $A$ is anti-complete with $B$ that there are only two possibilities for the restriction of $A$ and $B$ to the $k^{th}$ block. Either we have that both the restriction of $A$ is the subpath of $L$ with endpoints $c_k^L$ and $c_{k+1}^L$ and the restriction of $B$ is the subpath of $R$ with endpoints $c_k^R$ and $c_{k+1}^R$ or both the restriction of $A$ is the induced path made up of $c_{k+1}^L$ along with $b^1_k$ and $b^1_k$'s two neighbors in $L$ and the restriction of $B$ is the induced path made up of $c_{k}^R$ along with $a^1_k$ and $a^1_k$'s two neighbors in $R$.

By the definition of a $k$-creature, no vertex of the $100$-creature can belong to $N[A] \cap N[B]$. Hence, by what was just shown, the restriction of the $100$-creature to the blocks numbered between and including $i$ and $j$ induced two disjoint paths. Since the $i-1^{th}$ and $j+1^{th}$ blocks are the only two other blocks that can contain vertices from both $A$ and $B$, it is now easy to see that $(A, B$, $\{x_1, x_2, \ldots, x_{100}\}$, $\{y_1, y_2, \ldots, y_{100}\})$ cannot be a $100$-creature.
\end{proof}

\section{$k$-Creature and $k$-Skinny-Ladder Induced Minor Free Graphs}\label{creature and ladder free graphs}

In this section we will provide all the lemmas needed for a proof of Theorem \ref{theorem 1} and conclude this section with said proof. We begin this section by stating some well known results. Corollary \ref{vc corollary} then shows that the neighborhood of a vertex $v$ of a $k$-creature free graph $G$ can intersect the minimal separators of $G$ that do not contain $v$ in at most $n^{k+1}$ different ways. Lemma \ref{domination or creature/ladder} shows that the minimal separators of graphs that are $k$-creature-free and do not contain a $k$-skinny-ladder as an induced minor can be dominated by a constant number vertices. Lemma \ref{quasi-poly minimal separators} then uses a branching algorithm to list all minimal separators of its input graph assuming the input graph satisfies certain properties and proves a bound on the number of minimal separators produced by this algorithm. This lemma is combined with Corollary \ref{vc corollary} and Lemma \ref{domination or creature/ladder} to give a proof of Theorem \ref{theorem 1}.
%to show that if  a graph is $k$-creature free and does not contain an induced $k$-skinny-ladder then this algorithm outputs all minimal separators of the input graph and $n^{f(k)\log(n)}$ for some function $f$. 
%The proof of Theorem \ref{theorem 1} follows easily from Lemma \ref{quasi-poly minimal separators}. 
Most of the work of this section goes into proving lemmas needed for the proof of Lemma \ref{domination or creature/ladder}.

%$P$ of $G$ along with $k$ additional paths $P_1, P_2, \ldots, P_k$ of $G$ where $P_i$'s endpoints are $v_i$ and $x_i$. Furthermore, for all $i$ with $1 \leq i \leq k$ we require that $N(v_i) \cap V(P) \neq \emptyset$, $P$ and $P_i-v_i$ are anti-complete, and for $i \neq j$ $P_i$ and $P_j$ are anti-complete. We say the vertices $x_1, x_2, \ldots, x_k$ are the endpoints of the $k$-half-quasi-ladder. If $X \subset V(G)$ and $x_i \in X$ for all $i$ with $1 \leq i \leq k$, then we say $G$ contains a $k$-half-ladder ending in $X$. Note that a $k$-half-quasi-ladder is almost the same as a $k$-half-ladder, but we drop the requirement that if $a$ and $b$ are two neighbors of some $v_i$ in $P$, then there is no $v_j$, $i \neq j$ such that $v_j$ has a neighbor in the subpath of $P$ with endpoint $a$ and $b$.
%Here we will collect some general purpose lemmas that will be useful in proving Theorems \ref{theorem 1}, \ref{theorem 2}, and \ref{theorem 3}. 

\begin{lemma}[\bf Ramsey's Theorem]
\cite{doi:10.1112/plms/s2-30.1.264} 

For every pair of positive integer $k$ and $\ell$ there is a least positive integer $R(k,\ell)$ such that every graph with at least $R(k,\ell)$ vertices contains a clique of size $k$ or an independent set of size $\ell$.
\end{lemma}

Throughout this paper we will us the notation $R(k,\ell)$ to denote the least positive integer such that every graph with at least $R(k,\ell)$ vertices contains a clique of size $k$ or an independent set of size $\ell$.

\begin{lemma}[\bf Erdos-Szekers Theorem]
\cite{classicpapers}

For every pair on positive integers $r$ and $s$, any sequence of distinct real numbers of length at least ($r$-1)($s$-1) + 1 contains a monotone increasing subsequence of length $r$ or a monotone decreasing subsequence of length $s$.
\end{lemma}

\begin{definition}[\bf V.C. Dimension]
Let ${\cal F} = \{S_1, S_2,\ldots\}$ be a family of sets and let $H$ be a set. ${\cal F}$ is said to $shatter$ $H$ if for every subset $H' \subseteq H$ there is a $S_i \in {\cal F}$ such that $H' = S_i \cap H$. The $V.C.$-$dimension$ of ${\cal F}$ is the cardinality of the largest set that it shatters.
\end{definition}

\begin{lemma}[\bf Sauer-Shelah Lemma] \cite{SAUER1972145}

Let ${\cal F}$ be a family of sets such that the V.C.-dimension of ${\cal F}$ is $k$, and let $n = |\bigcup_{S_i \in {\cal F}} S_i|$, so $n$ is the number of distinct elements contained in the sets of ${\cal F}$. Then the number of sets of ${\cal F}$ is at most $\Sigma_{i=0}^k$ $n \choose i$ $\leq n^k$ + 1. 
\end{lemma}

Given a graph $G$ and two non adjacent vertices $u,v \in G$, we say that a $u,v$-minimal separator $S$ is $close$ to $v$ if $S \subseteq N(v)$. The following two lemmas establish useful properties of $u,v$-minimal separators close to $v$.

\begin{lemma}\label{close separator} \cite{KloksK94}
Given a graph $G$ and two non adjacent vertices $u,v \in V(G)$, there exists a unique $u,v$-minimal separator that is close to $v$. 
\end{lemma}

Let $S$ and $S'$ be two $u,v$ minimal separators for a graph $G$. Let $C_u$ be the connected component $u$ belongs to in $G-S$ and let $C_u'$ be the connected component $u$ belongs to in $G-S'$. Then we say that $S' \leq_{u,v} S$  if $C_u' \subset C_u$.

\begin{lemma}\label{close separator ordering}
Let $S_{v,u}$ be the $u,v$-minimal separator close to $v$ given by Lemma \ref{close separator} for some graph $G$ and let $S$ be another $u,v$-minimal separator. Then $S \leq_{u,v} S_{v,u}$.
\end{lemma}

\begin{proof}
Let $S_{v,u}$ be the $u,v$-minimal separator close to $v$ given by Lemma \ref{close separator} for some graph $G$ and let $S$ be another $u,v$-minimal separator. Assume for a contradiction that $S \leq_{u,v} S_{v,u}$ does not hold. Let $C_u$ be the connected component $u$ belongs to in $G-S_{v,u}$ and let $C_u'$ be the connected component $u$ belongs to in $G-S$. Then there is some vertex $x \in C_u'$ that is not in $C_u$. Since $x$ is not in $C_u$ it follows that every path from $u$ to $x$ must contain a vertex from $S_{v,u}$, but this implies that there must then be some vertex from $S_{v,u}$ that belongs to $C_u'$ since $x \in C'_u$ which means that $S$ would not separate $v$ from $u$ (since every vertex in $S_{v,u}$ is neighbors with $v$), a contradiction.
\end{proof}

%Recall that we will always assume that the graphs we consider contain a 1-creature. Every graph with at least one minimal separator contains a 1-creature, so this is without loss of generality.
\begin{lemma}\label{vc lemma}
If $S$ is a $u,v$-minimal separator for a graph $G$ that is $k$-creature free and ${\cal S}^+$ = $\{S' \cap S : S' \leq_{u,v} S, S'$ is a $u,v$-minimal separator$\}$, then $|{\cal S}^+| \leq |V(G)|^{k}$. 

%Similarly, if ${\cal S}^-$ = $\{S \cap S' : S \geq_{u,v} S', S'$ is a $u,v$-minimal separator$\}$, then $|{\cal S}^-| \leq |V(G)|^k$.
\end{lemma}

\begin{proof}
Let $S$ be a $u,v$-minimal separator of a $k$-creature free graph $G$, let $H$ be the $S$-full component of $G-S$ that contains $u$, and let $|S^+|$ be as in the statement of the lemma. We show that the V.C. dimension of $|S^+|$ is less than $k$. It will then follow by the Sauer-Shelah Lemma that $|S^+| \leq |V(G)|^{k-1} + 1 \leq |V(G)|^{k}$ (assuming $|V(G)| > 1$, but the lemma is trivially true if $|V(G)| = 1$ or $0$).

Assume for a contradiction that the V.C. dimension of $|S^+|$ is at least $k$. Then there exists minimal separators $S_1, S_2,\ldots, S_k$ and vertices $s_1, s_2, \ldots, s_k \in S$ such that $S \cap S_i = \{s_i\}$ and $S_i \leq_{u,v} S$. Let $H_i$ denoted the $S_i$-full component $u$ belongs to in $G-S_i$, let $s_i'$ be a neighbor of $s_i$ in $H_i$ such that its distance from $u$ in $H_i$ is minimum among all neighbors of $s_i$ in $H_i$, and let $P_i$ denote a shortest path from $s_i'$ to $u$ that is contained completely in $H_i$. It follows that the only vertex of $P_i$ that is neighbors with $s_i$ is $s_i'$. Let $P = \bigcup_{1 \leq i \leq k} P_i$, and let $K$ be the $S$-full component that contains $v$. We then have that ($V(K)$, $V(P)-\{s_1', s_2', \ldots, s_k'\}$, $\{s_1, s_2, \ldots, s_k\}$, $\{s_1', s_2' \ldots, s_k'\}$) forms a $k$-creature. To see this note that $P$ and $K$ are anti-complete since $S_i \leq_{u,v} S$, so all that must be verified is that if $s_is_j' \in E(G)$ then $i = j$ and that $P-\{s_1', s_2', \ldots, s_k'\}$ and $\{s_1, s_2, \ldots, s_k\}$ are anti-complete. We already saw before that $s_i'$ is the only vertex of $P_i$ that is neighbors with $s_i$. Furthermore, if there is some element $p_j \in P_j$ such that $s_ip_j \in E(G)$ but $i \neq j$ then this implies that $s_i \in H_j$ since $s_i \notin S_j$, but this would contradict the fact that $S_j \leq_{u,v} S$. It follows $P-\{s_1', s_2', \ldots, s_k'\}$ is disjoint from $\{s_1, s_2, \ldots, s_k\}$ and that $s_is_j' \in E(G)$ if and only if $i = j$.
\end{proof}

 As noted before, the following corollary is a key part of the proof of Theorem \ref{theorem 1}.
 
\begin{corollary}\label{vc corollary}
If $G$ is a $k$-creature-free graph, then for every $v \in V(G)$, with $S^v = \{N(v) \cap S : v \notin S$ and $S$ is a minimal separator of $G\}$, it holds that $|S^v| \leq |V(G)|^{k+1}$.
\end{corollary}

\begin{proof}
Let $G$ be a $k$-creature-free graph, and for every $v \in V(G)$ let $S^v$ be as in the statement of this lemma. For each $u$ in $G$ with $u \neq v$ and $uv \notin E(G)$, let $S_{v,u}$ be the $u,v$-minimal separator close to $v$ given by Lemma \ref{close separator}, and let $S^v_u = \{S_{v,u} \cap S : S$ is a $u,v$-minimal separator$\}$. Since for every $u,v$-minimal separator, $S'$, it holds that $S' \leq_{u,v} S_{v,u}$ by Lemma \ref{close separator ordering}, applying Lemma \ref{vc lemma} using $S_{v,u}$ as $S$, it follows that $S^v_u \leq |V(G)|^{k}$. We will show that $S^v = \bigcup S^v_u$ where the union is taken over all $u \in V(G)$ with $u \neq v$ and $uv \notin E(G)$. It will then follow that $|S^v| \leq |V(G)|^{k+1}$.

Let $S$ be some minimal separator such that $v \notin S$, let $C$ be the component $v$ is in in $G-S$, and let $u$ be a vertex in some $S$-full component of $G-S$ different from $C$, so $S$ is a $u,v$-minimal separator. Now, if $w \in N(v)$ and $w \in S$, then since $u$'s component of $G-S$ is $S$ full there is a path from $w$ to $u$ that contains no vertex of $C$ or $S$ other than $w$ and therefore contains no vertex of $N(v)$ other than $w$. So if $S_{v,u}$ did not contain $w$ then it would not separate $u$ from $v$, hence $w \in S_{v,u}$. Since $S_{v,u} \subset N(v)$ it follows that $S_{v,u} \cap S = N(v) \cap S$. Hence, $N(v) \cap S \in S^v_u$, and the result follows. 
\end{proof}

The following corollary will be needed in Sections \ref{cycle free graphs} and \ref{bounded clique}
\begin{corollary}\label{outside domination}
If $G$ is a $k$-creature-free graph and every minimal separator, $S$, of $G$ can be dominated by $k$ vertices of $G$ not in $S$, then $G$ has at most $|V(G)|^{k^2+2k}$ minimal separators.
\end{corollary}

\begin{proof}
Assume $G$ is a $k$-creature-free graph and every minimal separator, $S$, of $G$ can be dominated by $k$ vertices of $G$ not in $S$. For every $v \in V(G)$ let $S^v = \{N(v) \cap S : v \notin S$ and $S$ is a minimal separator of $G\}$. By Corollary \ref{vc corollary} it holds that $|S^v| \leq |V(G)|^{k+1}$. Let $X = \bigcup_{v \in G}S^v$. Then $|X| = |V(G)|^{k+2}$ and the assumption that all minimal separators, $S$, of $G$ can be dominated by $k$ vertices in $G$ not in $S$ implies that $S$ is the union of at most $k$ sets in $X$. It follows there are at most $|V(G)|^{k^2+2k}$ minimal separators in $G$.
\end{proof}

We remark that it is possible to generalize Corollaries \ref{vc corollary} and \ref{outside domination} to the $r^{th}$ neighborhood of a vertex for any fixed positive integer $r$ while still maintaining polynomial bounds by using the fact the family of $k$-creature-free graphs are closed under contracting edges.

The following lemmas will be building towards a proof of Lemma \ref{domination or creature/ladder}. We begin with an easy observation that will be useful in the proof of Lemma \ref{domination or creature/ladder}.
%By the independent set of a directed graph we mean a set of vertices of the graph such that no vertex in this set is an in-neighbor or out-neighbor of any other vertex in this set.

\begin{lemma}\label{directed independent set}
Let $G$ be a directed graph with maximum out-degree or maximum in-degree at most $c$, $c > 0$. Then $G$ has an independent set (no vertex is an in-neighbor or out-neighbor of any other vertex in this set) of size at least  $\frac{|V(G)|}{2c+1}$. Furthermore, if $|V(G)| \geq 2t$ and the maximum out-degree or maximum in-degree of $G$ is at most $\frac{1}{4t}|V(G)|$, then $G$ has an independent set of size at least $t$.  
\end{lemma}

\begin{proof}
Let $G$ be a directed graph. We will prove the statements for bounded maximum out-degree (for maximum in-degree the proof is nearly identical). If the maximum out-degree of $G$ is $c$, $c > 0$, then as long as $G$ has at least one vertex, there must exists a vertex $v \in G$ with in-degree at most $c$. If we let $G'$ be the subgraph induced by all vertices of $G-v$ that do not have $v$ as an in-neighbor or an out-neighbor, then the size of $G'$ is at least $|V(G)|$ - $2c - 1$, and $G'$ has maximum degree $c$. It follows by an inductive argument that we can find an independent set of size at least $\frac{|V(G)|}{2c+1}$.

To prove the furthermore statement, assume the maximum out degree of $G$ is at most $\frac{1}{4t}|V(G)|$, $t > 0$, and $|V(G)| \geq 2t$, so we have that $\frac{|V(G)|}{2t}+1$ $\leq$ $\frac{|V(G)|}{t}$. From the first paragraph we have that $G$ contains an independent set of size at least $\frac{|V(G)|}{\frac{2|V(G)|}{4t}+1} = \frac{|V(G)|}{\frac{|V(G)|}{2t}+1} \geq \frac{|V(G)|}{\frac{|V(G)|}{t}} = t$.
%A symmetric argument shows that if $G$ is a graph with maximum in-degree $c$, then $G$ has an independent set of size $\frac{|V(G)|}{2c+1}$.
\end{proof}

%\begin{lemma}\cite{10.1137/0404010}
%Let $G$ be a connected graph with degree at least 3. Then $G$ contains a spanning tree with at least $\frac{|V(G)|}{4}$ + 2 leaves.
%\end{lemma}

%\begin{corollary} \label{few leaves}
%Let $G$ be a connected graph that does not contain a spanning tree with at least $k$ leaves. Then $G$ is the union of at most $k^2$ induced paths of $G$.
%\end{corollary}

%\begin{proof}
%Let $G'$ be the graph attained by removing all vertices of degree two in $G$ adding an edge between any two vertices $x,y$ of $G$ of degree different from two that have an induced path in between them such that all internal vertices are of degree two, so $G$ can be attained from $G'$ by subdividing edges. Now let $G''$ be 
%Let $G$ be a connected graph that does not contain a spanning tree with at least $k$ leaves. It follows that there are less than $k$ leaves in $G$. Remove all leaves of $G$ and call the resulting Continue doing this until there are no more 

%\todo{there might be something to cite that proves this}
%\end{proof}

Lemmas \ref{path domination}, \ref{path domination helper}, and \ref{path domination helper 2} are used to help in the proof of Lemma \ref{path domination of nice set}. Lemma \ref{path domination of nice set} then is used to produce a structure that is similar to a $k$-skinny-ladder in graphs that are $k$-creature free and have minimal separators that cannot be dominated by few vertices. The structure that Lemma \ref{path domination of nice set} produces is then used in Lemma \ref{domination or creature/ladder} to produce a $k$-skinny-ladder as an induced minor in graphs that are $k$-creature-free and have minimal separators that cannot be dominated by few vertices.

\begin{lemma}\label{path domination}
Let $G$ be a graph that is $k$-creature free, let $S$ be a minimal separator of $G$, and let $A$ be an $S$-full component of $G-S$. Then $S$ is dominated by the union of less than $k$ induced paths in $A$.
\end{lemma}

\begin{proof}
Let $G$, $S$, and $A$ be as in the statement of this lemma, and let $A'$ be a minimally connected induced subgraph of $A$ such that $S$ is dominated by $A'$. Let $T$ be a breadth first search tree of $A'$ rooted at some vertex $v \in A'$, and let $L = \{\ell_1, \ell_2,\ldots,\ell_c\}$ be the set of leaves of $T$. Since $A'$ is minimal each leaf, $\ell_i \in L$, must have a neighbor $s_i \in S$ such that no other vertex of $A'$ is neighbors with $s_i$, else $A'-s_i$ would still be connected and dominate $S$. Then if $K$ is another $S$-full component different from $A$ the tuple $(V(A')-L, V(K), \{\ell_1, \ell_2,\ldots,\ell_c\}, \{s_1, s_2,\ldots, s_c\})$ forms a $c$-creature. It follows that if $G$ is $k$-creature-free, then T has at most $k-1$ leaves. It then follows that $A'$ is the union of at most $k-1$ induced paths and the result follows.
\end{proof}

In the following lemma we study a process takes three disjoint subgraphs/subsets $S,H,P$ of a $k$-creature-free graph $G$ and a vertex $v \in G$ and finds a small set $X$ such that no vertex of $S-N[X]$ shares a common neighbor with $v$ in $P$ or is neighbors with $v$. Note that in Lemma \ref{path domination helper} $v$ may or may not be in $S$.

\begin{lemma}\label{path domination helper}
Let $G$ be a graph and ($G, S, H, P, v$) be a tuple with the following properties: $S \subset V(G)$, $v \in G$, $H$ and $P$ are induced subgraphs of $G$ such that $H$ is connected, $P$ is an induced path, $S \cup \{v\}$, $V(H)$, $V(P)$ are mutually disjoint, $H$ is anti-complete with $P$, $v$ has no neighbor in $H$, and all vertices in $S-\{v\}$ have a neighbor in $H$. Then if $G$ does not contain a $k$-creature, there is a set, $X \subset S \cup V(H) \cup V(P)$, of size at most $2k-1$ such that $N(S-N[X]) \cap N(v) \cap V(P) = \emptyset$, and no vertex of $S-N[X]$ is neighbors with $v$.
\end{lemma}

\begin{proof}
Let $G$, $S$, $H$, $P$, and $v$ be as in the statement of this lemma. Number the vertices of $P$ 1 through $|V(P)|$ so that the vertex numbered $i$ is neighbors with the vertices numbers $i-1$ and $i+1$. We now consider the following process to build the set desired set $X$ such that $N(S-N[X]) \cap N(v) \cap V(P) = \emptyset$. 

We do the following for the first step of the process. Let $X_1 = \{v\}$, and let $S_1 = \{s : s\in S-N[X_1]$ and $N(s) \cap N(v) \cap V(P) \neq \emptyset\}$ (i.e. $S_1$ is the set of vertices of $S-N[v]$ that share a neighbor with $v$ in $P$). Label the vertices of $S_1$ by the lowest numbered vertex it is neighbors with in $V(P) \cap N(v)$. Let $s_1$ be a highest labeled vertex in $S_1$, and let $p_1$ be $s_1$'s lowest numbered neighbor in $N(v) \cap V(P)$. This completes the first step.

%For the second step we do the following. Let $X_2 = X_1 \cup \{s_1, p_1\}$, and let $S_2 = S_1 - N(X_2)$ and label the vertices of $S_2$ by the lowest numbered neighbor it has in $V(P) \cap N(v)$ (the vertices of $S_2$ inherits their label from their labels in $S_1$). Let $s_2$ be a highest labeled vertex in $S_2$ and let $p_2$ be $s_2$'s lowest numbered neighbor in $N(v) \cap V(P)$. Note by how we selected $v$, $s_1$, $p_1$, $s_2$, $p_2$ we have that among these five vertices $s_1$ is only neighbors with $p_1$, $s_2$ is only neighbors with $p_2$, and $v$ is only neighbors with $p_1$ and $p_2$.

For the i$^{th}$ step, $i > 1$, we do the following. Let $X_i = X_{i-1} \cup \{s_{i-1}, p_{i-1}\}$, and let $S_i = S - N[X_i]$ and label the vertices of $S_i$ by the lowest vertex it sees in $V(P) \cap N(v)$ (the vertices of $S_i$ inherit their labels from $S_{i-1}$). Let $s_i$ be a highest labeled vertex in $S_i$ and let $p_i$ be $s_i$'s lowest neighbor in $N(v) \cap V(P)$. This completes the $i^{th}$ step. Note by how we selected $v$, $s_1$, $p_1$, $s_2$, $p_2$, $\ldots$ $s_i$, $p_i$ that for $1 \leq a,b \leq i$, $s_a$ cannot be neighbors with $p_b$ if $a > b$ since $p_b$ would be in $X_a$ and therefore $s_a$ would not be in $S_a$, and $s_a$ cannot be neighbors with $p_b$ if $a < b$ since that would contradict either $p_a$ being $s_a$'s lowest numbered neighbor in $N(v) \cap P$ or $s_a$ being a highest labeled vertex in $S_a$. Hence, we then have that among these vertices $s_j$ is only neighbors with $p_j$ for $1 \leq j \leq i$, and $v$ is only neighbors with $p_j$ for $1 \leq j \leq i$.

%Note by how we selected $v$, $s_1$, $p_1$, $s_2$, $p_2$, $\ldots$ $s_i$, $p_i$ we then have that among these vertices $s_j$ is neighbors with $p_j$ and not $p_r$ for $r \neq j$, and $v$ is only neighbors with $p_j$ for $1 \leq j \leq i$.

We continue this process until we reach an $S_j$ that is empty. We claim this process cannot go past the $k^{th}$ step if $G$ does not contain a $k$-creature. Assume for a contradiction, that this process completes the $k^{th}$ step. We claim the tuple $(\{v\}, V(H), \{p_1, p_2,\ldots, p_k\}, \{s_1, s_2,\ldots, s_k\})$,  forms a $k$-creature.
%where the set $A$ corresponds to $\{v\}$, the set $B$  corresponds to $H$, the set $\{x_1, x_2,\dots, x_k\}$ corresponds to  $\{p_1, p_2,\ldots, p_k\}$, and the set $\{y_1, y_2,\ldots, y_k\}$ corresponds to $\{s_1, s_2,\ldots, s_k\}$. 
As noted before, by how we selected $v$, $s_1$, $p_1$, $s_2$, $p_2$, $\ldots$ $s_k$, $p_k$ we have that among these vertices $s_j$ is neighbors with $p_j$ and not with $p_r$ for $r \neq j$, and $v$ is only neighbors with the $p_j$'s. We also have that by assumption $v$ has no neighbors in $H$, but all of the vertices $s_1$, $s_2, \ldots, s_k$ have neighbors in $H$. Lastly, we can see that no vertex of $p_1, p_2, \ldots, p_k$ has a neighbor in $H$ by the assumption that $P$ is anti-complete with $H$. It follows that this tuple is a $k$-creature.

Set $X$ to be $X_j$, where $j$ is the first iteration where $S_j$ is empty. Since $S_j$ is empty, it follows $N(S-N[X]) \cap N(v) \cap V(P) = \emptyset$. We also have that no vertex of $S-N[X]$ is neighbors with $v$ since $v \in X$ and $|X| \leq 2k-1$ since $j \leq k$ and since the first step adds a single vertex and each iteration after that only adds two vertices.
\end{proof}

Note that in Lemma \ref{path domination helper 2}, $v$ may or may not be in $S$.

\begin{lemma}\label{path domination helper 2}
Let $G$ be a $k$-creature free graph and let ($G, S, H, P, v)$ be a tuple with the follow properties: $S \subset V(G)$, $v \in G$, $H$ and $P$ are connected induced subgraphs of $G$, $S \cup \{v\}$, $V(H)$, and $V(P)$ are vertex disjoint, $H$ is anti-complete with $P$, $v$ has no neighbor in $H$ or $S$, all vertices in $S-\{v\}$ have a neighbor in $H$ and in $P$, and for all $x \in S-\{v\}$ it holds that $N(x) \cap N(v) \cap V(P) = \emptyset$. Then there is a set of at most $k-1$ connected components of $P-N(v)$ such that every vertex of $S-\{v\}$ has a neighbor in at least one of these connected components.
\end{lemma}

\begin{proof}
Let $G, S, H, P$, and $v$ be as in the statement of this lemma. Assume for a contradiction that there does not exists a set of at most $k-1$ connected components of $P-N(v)$ such that every vertex of $S-\{v\}$ has a neighbor in at least one of these connected components. It follows then there is a set of $k$ connected components of $P-N(v)$, say $C_1$, $C_2, \ldots, C_k$, such that there exists $s_1$, $s_2, \ldots, s_k$ in $S$ where $N(s_i) \cap V(C_j) \neq \emptyset$ if and only if $i = j$. Since $P$ is connected, for every $C_i$ there exists a vertex $c_i \in N(v) \cap V(P)$ such that $c_i \in N(C_i)$ (the $c_i$'s may not be unique). Now, for each $s_i$, let $s_i'$ be a vertex in $C_i$ that $s_i$ is neighbors with such that there exists an induced path $P_i$, with internal vertices in $C_i$, from $s_i'$ to $c_i$ such that $s_i'$ is the only neighbor of $s_i$ on the path $P_i$. Then the tuple $(\{v\} \bigcup V(P_i-s_i'), H, \{s_1', s_2' \ldots, s_k'\}, \{s_1, s_2 \ldots, s_k\})$ is a $k$-creature,
%where the set $A$ corresponds to $\{v\} \cup (\bigcup P_i - s_i')$, $B$ corresponds to $H$, $\{x_1, x_2, \ldots, x_k\}$ corresponds to $\{s_1', s_2', \ldots, s_k'\}$, and $\{y_1, y_2, \ldots, y_k\}$ corresponds to $\{s_1, s_2, \ldots, s_k\}$
contradicting the assumption $G$ is $k$-creature-free.
\end{proof}

The following lemma produces a structure similar to a $k$-skinny-ladder in graphs that are $k$-creature-free and contain minimal separators that cannot be dominated by few vertices. This structure will be the main object of study in Lemma \ref{domination or creature/ladder}.

\begin{lemma}\label{path domination of nice set}
Let $S$ be a minimal separator of a graph $G$ such that $S$ cannot be dominated by $4k^5$ vertices. Then if $G$ is $k$-creature there exists there exists a subset $S'$ of $S$ of size $k$ such that there exists two paths, $P_1$ and $P_2$, in two different components of $G-S$ that dominate the vertices of $S'$, and no vertex of $P_1$ or $P_2$ has more than one neighbor in $S'$.
\end{lemma}

\begin{proof}
Assume that $G$ is a $k$-creature free graph, and let $S'$ be a minimal separator of $G$ that cannot be dominated by $4k^5$ vertices of $G$, and let $L'$ and $R'$ be two different $S'$-full components of $G$. It follows from Lemma \ref{path domination} that there is a set of less than $k$ induced paths in $L'$ that together dominated $S'$ and there is a set of less than $k$ induced paths in $R'$ that together dominate $S'$. It follows there exists two induced paths $L \subset L'$ and $R \subset R'$ such that $(N(L) \cap S'') \cap (N(R)) \cap S')$ cannot be dominated by $4k^3$ vertices of $G$. Let $S = (N(L) \cap S') \cap (N(R) \cap S')$. Fix a numbering the vertices of $L$ 1 through $|V(L)|$ so that the vertex numbered $i$ is neighbors with the vertices numbers $i-1$ and $i+1$.

Assume that we have an independent set of vertices $S_{i-1}$ of size $i-1$, $i \leq k$, and a vertex set $Z_{i-1}$ of size at most $4k^2(i-1)$, with the properties that no vertex $S-N[Z_{i-1}]$ is neighbors with a vertex in $S_{i-1}$, and any vertex in $L$ or $R$ that is neighbors with some vertex in $S_{i-1}$ has no other neighbors in $S_{i-1}$ or in $S-N[Z_{i-1}]$. We will show how to produce a set $S_i$ of size $i$ and $Z_i$ of size at most $4k^2i$ with the corresponding properties, assuming $i \leq k$. Note that the empty set satisfies the condition of $S_0$.

Let $S' = S-N[Z_{i-1}]$, and label the vertices of $S'$ according to the lowest numbered neighbor it has in $L$. Let $s$ be a highest labeled vertex in $S'$, since $i \leq k$ and since $S$ cannot be dominated by $4k^3$ vertices such an $s$ must exists. Let $\ell$ be the lowest numbered neighbor $s$ has in $L$ and assume the number of $\ell$ is $p$, and let $H$ denote the subpath of $L$ that is made up of the vertices labeled $1$ through $p-1$. We can then apply Lemma \ref{path domination helper} using ($G, S'-N(\ell), H, R, s$) to get a set $X'$ of size at most $2k-1$ such that $(S'-N[X \cup \{\ell\}]) \cap N(s) \cap V(R) = \emptyset$ and no vertex of $S'-N[X]$ is neighbors with $s$. Set $X = X' \cup \{\ell\}$.

We now wish to find a set $Y$ of size less than $2k^2$ so that no vertex of $S'-(N[X]\cup N[Y])$ shares a neighbor with $s$ in either $L$ or $R$. We first use Lemma \ref{path domination helper 2} on ($G, S'-N[X], H, R, s$) to get connected components $C_1, C_2, \ldots, C_c$, $c < k$, of $R-N(s$) such that all vertices of $S'-N[X]$ have a neighbor in at least one $C_i$. Then for each $C_i$ we apply Lemma \ref{path domination helper} on ($G$, $(S'-N[X]) \cap N(C_i)$, $C_i$, $L$, $s$) to get a set $Y_i$ of size less than $2k$ such that no vertex of $[(S'-N[X]) \cap V(C_i)]-N[Y_i]$ shares a neighbor with $s$ in $L$ (or $R$). It follows that if we set $Y$ = $\bigcup Y_i$ that no vertex of $(S'-N[X])-N[Y]$ shares a neighbor with $s$ in $L$ (or $R$). We may then set $S_i$ = $S_{i-1} \cup \{s\}$ and $Z_{i} = Z_{i-1} \cup X \cup Y$. Since $|X| \leq 2k$ and $|Y| \leq 2k^2$ we have that $|Z_i| \leq |Z_{i-1}| + 2k + 2k^2 \leq 4k^2i$.

The statement of the lemma now follows from the fact that $S$ cannot be dominated by $4k^3$ vertices so this process may go on until we attain the set $S_k$, which is the desired set, along with the paths $L$ and $R$.

\end{proof}

The next two Lemmas will be useful in the proof of Lemma \ref{domination or creature/ladder}.

\begin{lemma}\label{almost skinny implies skinny}
Let $G$ be a graph that contains a $k^2$-almost-skinny-ladder as an induced subgraph. Then $G$ contains a $k$-skinny-ladder as an induced minor.
\end{lemma}

\begin{proof}
Let $G$ be a graph that has a $k^2$-almost-skinny-ladder, $H$, as an induced subgraph. $V(H) = L \cup S \cup R$ where $L, S, R$  each have the same meaning as in the definition of an almost-skinny-ladder given in Section \ref{sec:prelim}. Number the vertices of $L$ 1 through $|V(L)|$ so that the vertex numbered $i$ is neighbors with the vertices numbered $i-1$ and $i+1$, and similarly, number the vertices of $R$ 1 through $|V(R)|$ so that the vertex numbered $i$ is neighbors with the vertices numbered $i-1$ and $i+1$.

Next we label each vertex in $S$ with a subscript 1 through $|S|$ so that for all $s_i, s_j \in S$ $i > j$ if and only if all of $s_i$'s neighbors in $L$ have a higher number than all of $s_j$'s neighbors in $L$ (by the definition of an almost-skinny-ladder such a numbering exists). Let $n(s_i)$ be the number of the highest numbered neighbor $s_i$ has is $R$. We now apply the Erodos-Szekers Theorem to the sequence $n(s_1), n(s_2)\ldots, n(s_{k^2})$ to get an increasing or decreasing subsequence of length at least $k$ and set $S^*$ to be the subset of $S$ that corresponds to the subsequence obtained from the Erodos-Szekers Theorem. If the Erodos-Szekers Theorem returned a decreasing subsequence then reverse the numbering of $R$, else leave it unchanged. Then for every $s_i, s_j \in S^*$, if $i > j$ then all of $s_i$'s neighbors in $L$ have a higher number than all of $s_j$'s neighbors in $L$ and  all of $s_i$'s neighbors in $R$ have a higher number than all of $s_j$'s neighbors in $R$. We can now apply the obvious edge contractions to $L$ and $R$ to form a $k$-skinny-ladder.
\end{proof}

\begin{lemma}\label{many dominating sets}
Let $G$ be a graph, let $a,b \in G$ be two non adjacent vertices of $G$, and let $P_1, P_2,\ldots,P_k$ be $K$ paths that are anti-complete with respect to one another and for every $P_i$, both $a$ and $b$ have a neighbor in $P_i$. Furthermore assume that for every $P_i$ that no vertex of $P_i$ is neighbors with both $a$ and $b$. Then $G$ contains a $k$-theta.
\end{lemma}

\begin{proof}
Let $G$, $a,b$, $P_1, P_2,\ldots,P_k$ be as in the statement of the lemma. For each $P_i$ we can then, by assumption, find a subpath of $P_i$, call it $P_i^*$, such that $P_i^*$ has endpoints $a_i, b_i$ where $a_i$ is neighbors with $a$, $b_i$ is neighbors with $b$, no internal vertex is neighbors with $a$ or $b$. It follow again by assumption that each $P_i^*$ has length at least 2 and that together the $P_i^*$'s along with $a$ and $b$ make a $k$-theta.
\end{proof}

The following lemma essentially takes a $k$-creature-free graph $G$ that has a minimal separator that cannot be dominated by few vertices, obtains the structure given by Lemma \ref{path domination of nice set} and cleans it up to produce $k$-skinny-ladder as an induced minor of $G$. 
%If $G$ is a $k$-creature free graph that cannot be dominated by few vertices, then the following lemma will essentially take the structure obtained by Lemma \ref{path domination of nice set} and clean it up to produce $k$-skinny-ladder as an induced minor of $G$.

\begin{lemma}\label{domination or creature/ladder} 
Let $S$ be a minimal separator of a graph $G$ such that $S$ cannot be dominated by $4[(8k^2)^{k+1}]^5$ vertices. If $G$ is $k$-creature-free, then $G$ contains a $k$-ladder as an induced minor.
\end{lemma}

\begin{proof}
Assume that $G$ is $k$-creature-free and $S'$ is a minimal separator of $G$ such that $S'$ cannot be dominated by $4[(8k^2)^{k+1}]^5$ vertices. It follows from Lemma \ref{path domination of nice set} that there is a set $S \subset S'$ of $(8k^2)^{k+1}$ vertices and two paths $R$ and $L$ that dominate $S$, $L$ anti-complete with $R$, and every vertex in $v \in V(L) \cup V(R)$ has at most one neighbor in $S$. 

Number the vertices of $L$ 1 through $|V(L)|$ so that the vertex numbered $i$ is neighbors with the vertices numbered $i-1$ and $i+1$, and number the vertices of $R$ 1 through $|V(R)|$ so that the vertex numbered $i$ is neighbors with the vertices numbered $i-1$ and $i+1$. For a vertex $x$ in $L$ or $R$ we will use the notation $n(x)$ to denote the number it has been given in $L$ or $R$. For every $s_j \in S$ let $\ell_j \in L$ and $r_j \in R$ be the highest numbered neighbors of $s_j$ in $L$ and $R$ respectively. We now set $L_1 = L, R_1 = R$, and $S_1 = S$. We will consider the following process to produce a $k^2$-almost-skinny-ladder. We will show  this process cannot go past $k$ iterations if $G$ is $k$-creature-free, and we will ensure that at the $i^{th}$ step that $V(L_i) \subset V(L)$, $V(R_i) \subset V(R)$, $S_i \subset S$, $|S_i| \geq (8k^2)^{k-i+2}$, $L_i$ and $R_i$ are induced paths, and if $s_j \in S_i$ then $\ell_j \in L_i$ and $r_j \in R_i$. We will also produce induced subpaths $P_i$ of either $L$ or $R$ such that the $P_i$'s are anti-complete with respect to one another and the vertices of $P_i$ will dominate $S_j$ if $i < j$.

At the $i^{th}$ step, $i \leq k$, we do as follows. Create an auxiliary directed graph, $AUX_i$, whose vertex set is $S_i$ and there is an edge from $s_a \in S_i$ to $s_b \in S_i$ if at least one of the following two cases hold

\begin{enumerate}

    \item $n(\ell_a) > n(\ell_b)$ and $s_a$ has a neighbor $x$ in $L$ such that $n(x) < n(\ell_b)$
    
   % \item $n(\ell_a) < n(\ell_b)$ and $s_a$ has a neighbor $x$ in $L$ such that $n(x) > n(\ell_b)$
    
    \item $n(r_a) > n(r_b)$ and $s_a$ has a neighbor $x$ in $R$ such that $n(x) < n(r_b)$
        
%    \item $n(r_a) < n(r_b)$ and $s_a$ has a neighbor $x$ in $R$ such that $n(x) > n(r_b)$
    
\end{enumerate}

If the maximum in-degree of $AUX_i$ is at most $\frac{1}{4k^2}|S_i|$ then we stop. Since $|S_i| \geq (8k^2)^{k-i+2}$ this gives an independent set of size at least $k^2$ by Lemma \ref{directed independent set}. If there is an $s_j \in S_i$ with in degree over $\frac{1}{4k^2}|S_i|$ then for at least $\frac{1}{8k^2}$ fraction of the vertices of $S_i$, call this subset of vertices $S_{i+1}$, all vertices $s \in S_{i+1}$ must satisfy case 1 one with $s$ playing the role of $s_a$ and $s_j$ playing the role of $s_b$, or all vertices $s \in S_{i+1}$ must satisfy case 2 again with $s$ playing the role of $S_a$ and $s_j$ playing the role of $s_b$. For each case we now describe what to do if all the vertices of $S_{i+1}$ satisfy that case
%at least $\frac{1}{8k^2}$ fraction of the vertices of $S_i$ satisfy that case 
(if all vertices of $S_{i+1}$ satisfy both cases, then we go with the first case).
Each number here corresponds what to do in that case.
\begin{enumerate}

    \item Call $P_i$ the subpath of $L_i$ that is made up of vertices with numbers less than $n(\ell_j)$. Set $R_{i+1}$ = $R_i$ and set $L_{i+1}$ to be the vertices of $L$ with numbers greater than $n(\ell_j)$.

%    \item Call $P_i$ the subpath of $L_i$ that is made up of vertices with numbers greater than $n(\ell_j)$. Set $R_{i+1}$ = $R_i$ and set $L_{i+1}$ to be the vertices of $L$ with numbers less than $n(\ell_j)$.
    
    \item Call $P_i$ the subpath of $R_i$ that is made up of vertices with numbers less than $n(r_j)$. Set $L_{i+1}$ = $L_i$ and set $R_{i+1}$ to be the vertices of $R$ with numbers greater than $n(r_j)$.
    
%    \item Call $P_i$ the subpath of $R_i$ that is made up of vertices with numbers greater than $n(r_j)$. Set $R_{i+1}$ = $R_i$ and set $R_{i+1}$ to be the vertices of $R$ with numbers less than $n(r_j)$.
    
\end{enumerate}

It can then be seen that $V(L_{i+1}) \subset V(L)$, $V(R_{i+1}) \subset V(R)$, $S_{i+1} \subset S$, $|S_{i+1}| \geq (8k^2)^{k-i+1}$, $L_{i+1}$ and $R_{i+1}$ are induced paths, and if $s_j \in S_{i+1}$ then $\ell_j \in L_{i+1}$ and $r_j \in R_{i+1}$ as required. Furthermore, it can be seen that any of the previously $P_j$'s that have been produced in this process ($j \leq i$) dominate all vertices of $S_{i+1}$ and are anti-complete with respect to one another. By Lemma \ref{many dominating sets} then, this process cannot go past the $k^{th}$ iteration without producing a $k$-theta.

We conclude there is some step $j \leq k$ such that the auxiliary graph $AUX_j$ has max in-degree less than $\frac{1}{4k^2}|S_j|$, and since $|S_j| \geq 8k^2$ it therefore has an independent set of size $k^2$ by Lemma \ref{directed independent set}. Let $S^*$ denote such an independent set, we claim that $G[V(L) \cup S^* \cup V(R)]$ makes an $k^2$-almost-skinny-ladder. Let $x,y \in S^*$ and let $a,b$ be the highest and lowest numbered neighbors of $x$ in $L$ respectively, and assume that $y$ has a neighbor $c$ on the induced path of $L$ that has $a$ and $b$ as its endpoints. If $y$'s highest numbered neighbor in $L$ is greater than $n(a)$ then $y$ has an edge to $x$ in $AUX_j$ by case (1). If $y$'s highest numbered neighbor is $L$ is less than $n(a)$, then $x$ has an edge to $y$ again by case (1). Both cases yield a contradiction to $S^*$ being an independent set in $AUX_j$.  A nearly identical argument show that if $a',b'$ are $x$'s highest and lowest numbered neighbors $R$ respectively, then $y$ cannot have a neighbor in the induced subpath of $R$ that has $a',b'$ as its endpoints. It follows that  $G[V(L) \cup S^* \cup V(R)]$ is a $k^2$-almost-skinny-ladder. Applying Lemma \ref{almost skinny implies skinny} shows that $G$ contains a $k$-skinny-ladder as an induced minor.

\end{proof}

The following lemma uses a branching algorithm to produce all of the minimal separators of a graph $G$ and proves a bound on the number of minimal separators produced by this algorithm, which when combined with Corollary \ref{vc corollary} and Lemma \ref{domination or creature/ladder} gives a proof of Theorem \ref{theorem 1}.

\begin{lemma}\label{quasi-poly minimal separators} 
There exists a function $f : \mathbb{N} \rightarrow \mathbb{N}$ such that the following holds. Let $G$ be a graph and let $k$ and $c$ be integers such that for all induced subgraphs $G'$ of $G$ and for all $v \in V(G')$, if $S^v_{G'} = \{N(v) \cap S : v \notin S$ and $S$ is a minimal separator of $G'\}$, then $|S^v_{G'}| \leq c$ and every minimal separator of any induced subgraph of $G$ can be dominated by $k$ vertices. Then $G$ has a most $(c+n^k)^{f(k)\log(n)}$ minimal separators where $n = |V(G)|$.
\end{lemma}

\begin{proof}
Let $G$, $S^v_G$, $k$, and $c$ be as in the statement of this lemma. The proof of the bound makes use of a branching algorithm. The algorithm takes as input $G$ and $X \subset V(G)$ and the algorithm will use the set $K_{ret}$ to store the vertex sets it will return. It will return $K_{ret}$ which will contain all minimal separators of $G$ contained in $X$ (most likely along with other vertex sets).  We have no concern about the runtime of the algorithm, but we care about the size of the final set it returns. The algorithm is intended to be used initially on the input ($G$, $V(G)$). 

Assume the the input to the algorithm is ($G, X$). If $X$ is empty, then the algorithm returns $\emptyset$. Else, the algorithm determines the set $Q \subset V(G)$ where $Q$ contains all vertices $v \in V(G)$ such that $|N[v] \cap X| \geq \frac{1}{2k}|X|$. Then the algorithm branches in the following two ways:

\begin{enumerate}

    \item For every $q \in Q$ and every $Y \in S^q_G$ the algorithm recursively calls itself on ($G-Y$, $X-N_G[q]$). Each recursive call returns a set $K'$, which contains vertex sets. Then if the recursive call ($G-Y$, $X-N_G[q]$) returns the collection $K'$ of vertex sets, for each set $S$ in $K'$, the algorithm adds the set $S \cup Y$ to $K_{ret}$. 
    
    \item For every set $R$ of $k$ vertices of $G$ such that $R \cap Q = \emptyset$, the algorithm recursively calls itself on ($G-Q$, $(X \cap N_G(R))-Q$). Each recursive call returns a set $K'$, which intern contains vertex sets. Then for each set, $S$, in each $K'$ returned the algorithm adds the set $S \cup Q$ to $K_{ret}$. 
    
\end{enumerate}

After completing this, the algorithm then returns the set $K_{ret}$. Note that in (2) since the set $R$ has no vertex in $Q$ and $|R| \leq k$, the neighborhood of $R$ contains at most $\frac{1}{2}$ of the vertices of $X$.

Since $Q$ contains all vertices $v \in V(G)$ such that $|N[v] \cap X| \geq \frac{1}{2k}|X|$, each recursive call the algorithm makes is on input ($G'$, $X'$) where $|X| \geq (1- \frac{1}{2k})|X'|$, so the algorithm terminates. Let $S$ be a minimal separator of $G$ contained in $X$. Assume all of the recursive calls ($G'$, $X'$) the algorithm makes returns a set that contains all minimal separators of $G'$ contained in $X'$. If $Y = N_G(q) \cap S$ for some $q \in G$ and $q \notin S$, then $S-Y$ is a minimal separator of $G-Y$ that is contained in $X-N_G[q]$. So if there is a $q \in Q$ such that $q \notin S$, then $S$ gets added to  $K_{ret}$ in (1). If $Q \subset S$, then $S-Q$ is a minimal separator of $G-Q$, and by assumption there exists some collection of at most $k$ vertices, $R$, in $G-Q$ such that $S-Q \subset N_{G-Q}(R)$ and therefore $S-Q \subset (X \cap N_G(R)) - Q$. It follows that in this case we also have $S$ gets added to $K_{ret}$ in (2). Induction on the the depth of the recursive call now shows that this algorithm returns all minimal separators.

If $T(n, x)$ represents the maximum number of minimal separators that a vertex set $X$ of size at most $x$ can contains for any graph $G$ with $|V(G)| \leq n$ and $X \subset V(G)$, such that the graph $G$ satisfies the conditions of the lemma, then the algorithm shows that $T(n, x) \leq (c+n^k) T(n,[1-\frac{1}{2k}]x)$. %Since $T(n,0) = 0$, 
Using the fact that $\lim_{y\rightarrow \infty}(1-\frac{1}{y})^y = \frac{1}{e}$ we expand the inequality $T(n, x) \leq (c+n^k) T(n,[1-\frac{1}{2k}]x)$ out $O(k)$ times to get $T(n, x) \leq (c+n^k)^{O(k)} T(n,\frac{1}{2}x)$. Since $T(n,0) = 0$ it follows that there exists a function $f : \mathbb{N} \rightarrow \mathbb{N}$ (independent of the choice of $k$ or $G$) such that this solves to $T(n,x) \leq (c+n^k)^{f(k)\log(x)}$. By taking the initial $X$ to be $V(G)$, it follows that $G$ then contains at most $(c+n^k)^{f(k)\log(n)}$ minimal separator, where $n = |V(G)|$.
\end{proof}

We are now ready to prove Theorem \ref{theorem 1}.

\begin{proof}[Proof of Theorem~\ref{theorem 1}]

Let $G$ be a graph, $|V(G)| = n$, that is $k$-creature-free and has no $k$-skinny-ladder as an induced minor. For every induced subgraph $G'$ of $G$ and for every $v \in G'$, let $S^v_{G'} = \{N(v) \cap S : v \notin S$ and $S$ is a minimal separator of $G'\}$. Then $|S^v_{G'}| = n^{f(k)}$ for some function $f : \mathbb{N} \rightarrow \mathbb{N}$ by Corollary \ref{vc corollary} ($f$ is independent of the choice of $k$ or $G$). By Lemma \ref{domination or creature/ladder}, since $G$ is $k$-creature free and has no $k$-skinny-ladder as an induced minor, there is a function $f' : \mathbb{N} \rightarrow \mathbb{N}$ ($f'$ is independent of the choice of $k$ or $G$) such that every minimal separator of any induced subgraph of $G'$ is dominated by $f'(k)$ vertices. Lemma \ref{quasi-poly minimal separators} then implies there is a function $f'' : \mathbb{N} \rightarrow \mathbb{N}$ ($f''$ is independent of the choice of $k$ or $G$) such that $G$ has at most $({f(k)} + n^{f'(k)})^{f''(k)\log(n)}$ minimal separators. We can then see there exists a function $f^*: \mathbb{N} \rightarrow \mathbb{N}$ ($f^*$ is independent of the choice of $k$ or $G$)
%to equal $f''(k)(f(k) + f'(k))$ then
such that $G$ has at most $n^{f^*(k)\log(n)}$ minimal separators. It follows that the family of graphs that are $k$-creature-free and do not contain a $k$-skinny-ladder as an induced minor are strongly-quasi-tame. 
\end{proof}

\section{Finite Forbidden Induced Subgraphs}\label{finite forbidden}

In this section we will provide the lemmas needed in the proof of Theorem \ref{theorem 2} as well give a proof of Theorem \ref{theorem 2} at the end of this section. The majority of the work of this section goes into proving that given an integer $k$, if $G$ contains a $k'$-creature for large enough $k'$, then $G$ must contain a $k$-theta, $k$-prism, $k$-pyramid, $k$-ladder-theta, $k$-ladder-prism, or $k$-ladder as an induced subgraph, which is proven in Lemma \ref{k creature implies theta}. Lemmas \ref{exponentially many mins seps} and \ref{k-claws are untame} then show that if ${\cal F}$ is a family of graphs defined by a finite number of forbidden induced subgraphs and ${\cal F}$ does not forbid all $k$-theta, $k$-prism, $k$-pyramid, $k$-ladder-theta, $k$-ladder-prism, $k$-claw, $k$-paw graphs for all $k$ larger than some fixed constant, then ${\cal F}$ is feral.
%%%%%%%%%%%%%%%%%%%%%%%%%%%%%%%%%%%%%%%%%%%%%there exists a constant $c > 1$ such that for every integer $N$ there is a graph $G \in {\cal F}$ such that $|V(G)| = n > N$ and $G$ has at least $c^n$ minimal separators.
Theorem \ref{theorem 2} is then proved using Lemma \ref{k creature implies theta} along with Theorem \ref{theorem 1} and a few simple observations, as well as Lemmas \ref{exponentially many mins seps} and \ref{k-claws are untame}.

It will be useful in this section to define the following graphs. These graphs are depicted in Figure \ref{half graphs}.

\begin{itemize}
    
    \item A graph $G$ is a $k$-$half$-$theta$ if $G$ consists of a vertex $v$ and $k$ induced paths $P_1$, $P_2, \ldots, P_k$ of $G$ such that each path has length at least 2, for $1 \leq i \leq k$ it holds that $v$ is one endpoint of $P_i$, and for $j \neq i$ it hold that $P_i-v$ is anti-complete with $P_j-v$. Let $x_i$ denote the endpoint of $P_i$ that is not $v$. Then we say the vertices $x_1, x_2, \ldots, x_k$ are the endpoints of the $k$-half-theta. If $X$ is a vertex set and $x_i \in X$ for all $i$ with $1 \leq i \leq k$, then we say $G$ is a $k$-half-theta ending in $X$. 

    \item A graph $G$ is a $k$-$half$-$prism$ if $G$ consists of a clique of vertices $v_1, v_2, \ldots, v_k$ and $k$ induced paths $P_1$, $P_2$,..., $P_k$ of $G$ such that each path has length at least 1, for $1 \leq i \leq k$ it holds that $v_i$ is one endpoint of $P_i$, and for $j \neq i$ it hold that $P_i-v_i$ is anti-complete with $P_j$. If the length of $P_i$ is greater than 1 then let $x_i$ denote the endpoint of $P_i$ that is not $v_i$, and if the length of $P_i$ is 1 then let $x_i = v_i$. We say the vertices $x_1, x_2, \ldots, x_k$ are the endpoints of the $k$-half-prism. If $X$ is a vertex set and $x_i \in X$ for all $i$ with $1 \leq i \leq k$, then we say $G$ is a $k$-half-prism ending in $X$. 
    
    \item A graph $G$ is a $k$-$half$-$ladder$ if $G$ consists of a path $P$ of $G$ along with $k$ additional paths $P_1, P_2, \ldots, P_k$ of $G$ such that each path has length at least 1. For $1 \leq i \leq k$ let $P_i$'s endpoints be $v_i$ and $x_i$ (with $v_i$ possibly equal to $x_i$). We call $P$ the backbone path and the $P_i$'s the auxiliary paths. We require that $v_i$ has at least one neighbor in $P$, $P$ is anti-complete with $P_i-v_i$, and for $j \neq i$ $P_i$ is anti-complete with $P_j$. Lastly, we also require that if $a$ and $b$ are two neighbors of some $v_i$ in $P$, then there is no $v_j$, $i \neq j$ such that $v_j$ has a neighbor in the induced subpath of $P$ with endpoint $a$ and $b$. We say the vertices $x_1, x_2, \ldots, x_k$ are the endpoints of the $k$-half-ladder. If $X$ is a vertex set and $x_i \in X$ for all $i$ with $1 \leq i \leq k$, then we say $G$ is a $k$-half-ladder ending in $X$. 
    
    \item A graph $G$ is a $k$-$half$-$quasi$-$ladder$ if $G$ consists of a path $P$ of $G$ along with $k$ additional paths $P_1, P_2, \ldots, P_k$ of $G$ such that each path has length at least 1. For $1 \leq i \leq k$ let $P_i$'s endpoints be $v_i$ and $x_i$ (with $v_i$ possibly equal to $x_i$). We call $P$ the backbone path and the $P_i$'s the auxiliary paths. We require that $v_i$ has at least one neighbor in $P$, $P$ is anti-complete with $P_i-v_i$, and for $j \neq i$ $P_i$ is anti-complete with $P_j$. We say the vertices $x_1, x_2, \ldots, x_k$ are the endpoints of the $k$-half-quasi-ladder. If $X$ is a vertex set and $x_i \in X$ for all $i$ with $1 \leq i \leq k$, then we say $G$ is a $k$-half-ladder ending in $X$. Note that a $k$-half-quasi-ladder is almost the same as a $k$-half-ladder, but we drop the requirement that if $a$ and $b$ are two neighbors of some $v_i$ in $P$, then there is no $v_j$, $i \neq j$ such that $v_j$ has a neighbor in the subpath of $P$ with endpoint $a$ and $b$.
    
\end{itemize} 

\begin{figure}
\centerline{\includegraphics{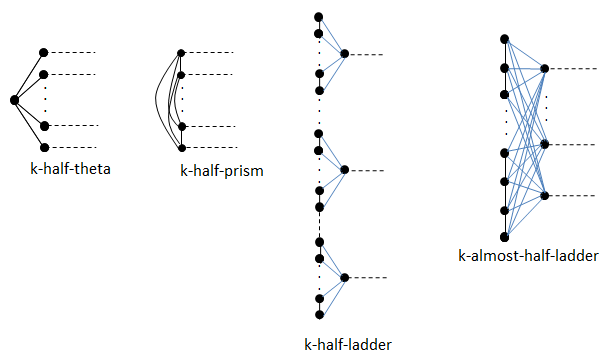}}
\caption{\em Dashed lines represent the option of having an arbitrary length path (possibly of length 0). The blue lines used in the $k$-half-ladder and $k$-almost-half-ladder graphs represents the option of either having or not having that edge, but for each vertex not on the backbone path that is adjacent at least one blue edges, at least one of those blue edges must belong to the graph.}
\label{half graphs}
\end{figure}

The following lemmas, culminating with Lemma \ref{k creature implies theta}, work towards proving that given an integer $k$, if $G$ contains a $k'$-creature for large enough $k'$, then $G$ must contain a $k$-theta, $k$-prism, $k$-pyramid, $k$-ladder-theta, $k$-ladder-prism, or $k$-ladder. 

Lemmas \ref{high degree} through \ref{final step} are used to prove Lemma \ref{final step 2}, which shows that if ($A$, $B$, $\{x_1, x_2,\ldots,x_{k'}\}$, $\{y_1, y_2,\ldots,y_{k'}\}$) is a $k'$ creature for large enough $k'$, then $G[A \cup \{x_1, x_2,\ldots,x_{R(k',k')}\}]$ contains an induced $k$-half-theta, $k$-half-prism, or a $k$-half-quasi-ladder, ending in $\{x_1, x_2,\ldots,x_{R(k',k')}\}$.

\begin{lemma}\label{high degree}
Let $G$ be a graph that contains a $k$-creature ($A$, $B$, $\{x_1, x_2,\ldots,x_{k}\}$, $\{y_1, y_2,\ldots,y_{k}\}$) where $\{x_1, x_2,\ldots,x_{k}\}$ is an independent set of $G$. Let $A'$ be a minimally connected induced subgraph of $G[A]$ such that $\{x_1, x_2,\ldots,x_{k}\} \subset N(A')$. If $A'$ contains a vertex with degree at least $R(d,d)$ in $A'$, then $G[A \cup \{x_1, x_2,\ldots,x_{k}\}]$ contains a $d$-half theta or a $d$-half-prism ending in $\{x_1, x_2,\ldots,x_{k}\}$.
\end{lemma}

\begin{proof}
Let $G$ be a graph that contains a $k$-creature ($A$, $B$, $\{x_1, x_2,\ldots,x_{k}\}$, $\{y_1, y_2,\ldots,y_{k}\}$). Let $A'$ be a minimally connected induced subgraph of $G[A]$ such that $\{x_1, x_2,\ldots,x_{k}\} \subset N_G(A')$. Assume $v \in A'$ has degree at least R$(d,d)$ in $A'$. Let $v_1, v_2,\ldots, v_{R(d,d)}$ be distinct neighbors of $v$ in $A'$. By the minimality of $A'$, for each $v_i$ there must be a vertex $x_{v_i}$ such that every path starting from $v$ and ending at $x_{v_i}$ with internal vertices contained $A'$ must contain $v_i$, since if this does not happen for some given $v_i$ then the connected component of $A' - v_i$ that contains $v$ would be a proper induced subgraph of $A'$ that is connected and whose open neighborhood contains $\{x_1, x_2,\ldots,x_{k}\}$. It follows there must exist induced paths $P_1, P_2,\ldots, P_{R(d,d)}$ such that $v_i \in P_i$, $P_i$'s endpoints are $v_i$ and $x_{v_i}$, and $P_i-v_i$ is anti-complete with $P_j$. 
%$(N[P_i-v] \cap V(P_j-v)) \cup (N[P_j-v] \cap V(P_i-v)) = \emptyset$ or $\{v_i, v_j\}$. 
We then apply Ramsey's Theorem to the $v_i$'s get a subset of size $d$ of the $P_i$'s that along with $v$ form a $d$-half theta that ends in $\{x_1, x_2,\ldots,x_{k}\}$ (if Ramsey's Theorem provides an independent set of size $d$) or a subset of size $d$ of the $P_i$'s that form a $d$-half prism that ends in $\{x_1, x_2,\ldots,x_{k}\}$ (if Ramsey's Theorem provides a clique of size $d$) and the result now follows.
\end{proof}

\begin{lemma}\label{finding path with many good vertices}
Let $G$ be connected graph with maximum degree $d$ and contains at least $d^k$ vertices with degree greater than 2. Then there exists an induced path of $G$ that contains at least $k$ vertices of degree greater than 2.
\end{lemma}

\begin{proof}
Let $G$ be a connected graph with maximum degree $d$ and contains at least $d^k$ vertices with degree greater than 2. Let $T$ be a breadth first search tree of $G$ rooted at some vertex $v \in G$. We create the desired path as follows. Let $v_1$ be the first descendent of $v$ in $T$ that has degree greater than 2 in $G$ ($v_1$ could be $v$). We begin our path at $v_1$. We will grow the path $P_i = \{x_1, x_2,\ldots x_m\}$ where $x_1 = v_1$, $x_j$ is the parent of $x_{j+1}$ in $T$, $P_i$ contains at least $i$ vertices of $G$ with degree greater than 2 in $G$, and the subtree of $T$ rooted at $x_m$ contains at least $d^{k-i+1}$ vertices of $G$ with degree greater than 2 in $G$. 

Assume that we have such a path $P_i = \{x_1, x_2,\ldots x_m\}$, $i < k$ (the vertex $v_1$ satisfies the conditions of $P_1$). We will show how to attain $P_{i+1}$. Since the maximum degree in $G$ is $d$, $x_m$ has at most $d$ children in $T$, and by assumption the subtree of $T$ rooted at $x_m$ has at least $d^{k-i+1}$ vertices of degree greater than 2 in $G$, it follows that for at least one child, call it $x_{m+1}$, the subtree rooted at $x_{m+1}$ has at least $d^{k-i}$ vertices of $G$ with degree greater than 2 in $G$. Now let $v_{i+1}$ be the first descendant of $x_{m+1}$ with degree different from 2 in $G$ ($v_{i+1}$ could be $x_{m+1}$) and let $P_{i+1}$ be the path $P_i$ along with the induced path in $T$ from $x_{m+1}$ to $v_{i+1}$. It follows $P_{i+1}$ satisfies the required conditions.

Hence we can produce a $P_k$ that satisfies the conditions stated before, and we can then see that $P_k$ is an induced path in $G$ with at least $k$ vertices of degree greater than 2.
\end{proof}

\begin{lemma}\label{using path with many good vertices}
Let $G$ be a graph that contains a $k$-creature ($A$, $B$, $\{x_1, x_2,\ldots,x_{k}\}$, $\{y_1, y_2,\ldots,y_{k}\}$) where $\{x_1, x_2,\ldots,x_{k}\}$ is an independent set. Let $A'$ be a minimally connected subgraph of $G[A]$ such that $\{x_1, x_2,\ldots,x_{k}\}$ $\subset$ $N(A')$. If $A'$ contains an induced path, $P$, with at least R$(d,d)$ vertices of degree greater than 2 in $A'$, then there is a $d$-half-quasi-ladder or a $d$-half-prism in $G[A \cup \{x_1, x_2,\ldots,x_{k}\}]$ that ends in $\{x_1, x_2,\ldots,x_{k}\}$.
\end{lemma}

\begin{proof}
Let $G$, $A'$, $\{x_1, x_2,\ldots,x_{k}\}$, and $P$ be as in the statement of the lemma, let $v_1, v_2, \ldots, v_{R(d,d)}$ be vertices of $P$ that have degree greater than 2 in $A'$, and for each $v_i$ let $v_i'$ be a neighbor of $v_i$ in $A'$ that is not in $P$. By the minimality of $A'$, for each $v_i'$ there must exist a vertex $x_{v_i}$ such that every path from $v_i$ to $x_{v_i}$ with internal vertices contains in $A'$ must contain $v_i'$, since if this does not happen for some given $v_i'$ then the component of $A'-v_i'$ that contains $v_i$ would be a proper induced subgraph of $A'$ that is connected and whose open neighborhood contains  $\{x_1, x_2,\ldots,x_{k}\}$. It follows there must exists induced paths $P_1, P_2,\ldots, P_{R(d,d)}$ disjoint from $P$ with internal vertices contained in $A'$, $P_i$'s endpoints are $v_i'$ and $x_{v_i}$, and $P_i-v_i'$ is anti-complete with $P_j$. We then apply Ramsey's Theorem to the $v_i'$'s to get a subset of size $d$ of the $P_i$'s along with $P$ that form a $d$-half-quasi-ladder that ends in $\{x_1, x_2,\ldots,x_{k}\}$ (if Ramsey's Theorem provides an independent set of size $d$) or a subset of size $d$ of the $P_i$'s that yield a $d$-half-prism that ends in $\{x_1, x_2,\ldots,x_{k}\}$ (if Ramsey's Theorem provides a clique of size $d$).
\end{proof}

\begin{lemma}\label{final step}
Let $G$ be a graph that contains a $k\cdot (d^{c+1}+d)$-creature ($A$, $B$, $\{x_1, x_2,\ldots,x_{k\cdot (d^{c+1}+d)}\}$, $\{y_1, y_2,\ldots,y_{k\cdot (d^{c+1}+d)}\}$). Let $A'$ be a minimally connected subgraph of $G[A]$ such that $\{x_1, x_2,\ldots,x_{k\cdot (d^{c+1}+d)}\}$ $\subset$ $N(A')$. Assume the max degree in $A'$ is $d$ and that $A'$ contains less than $d^c$ vertices of degree greater than 2 in $A'$. Then $G[A \cup \{x_1, x_2,\ldots,x_{k\cdot (d^{c+1}+d)}\}]$ contains a $k$-half-quasi-ladder ending in  $\{x_1, x_2,\ldots,x_{k\cdot (d^{c+1}+d)}\}$.
\end{lemma}

\begin{proof}
%Let $G$, $A'$ and $\{x_1, x_2,\ldots,x_{k\cdot d^{4c}}\}$ be as in the statement of the lemma. If $A'$ is a path then we are done since $A'$ along with $\{x_1, x_2,\ldots,x_{k}\}$ is a $k$-half-quasi-ladder, so assume $A'$ is not a path. Since $A'$ is not a path, the number of leaves of $A'$ is at most the max degree of $A'$ times the number of vertices with degree greater than 2. This gives us at most $d^{c+1}$ vertices of degree different from 2. If we take an induced We first show that $A'$ is the union of $d^{4c}$ induced paths in $A'$. Since the maximum degree in $A'$ is $d$, each vertex is neighbors there are at most $d^c$ + $d^{c+1}$ vertices of degree different from 2 in $A'$.

Let $G$, $A'$ and $\{x_1, x_2,\ldots,x_{k\cdot (d^{c+1}+d)}\}$ be as in the statement of the lemma. Let $T$ be a breadth first search tree of $A'$ rooted at some vertex $v$. Then $T$ is a tree in which every vertex except for the root can have at most $d-1$ children, hence there are at most $d^c+1$ vertices that have more than one descendent, and the maximum number of decedents any vertex from this set can have is $d$. It follows that there are at most $d^{c+1}+d$ leaves of $T$, and therefore $A'$ is the union of at most $d^{c+1}+d$ induced paths in $A'$. Hence, there exists some induced path $P$ in $A'$ such that $P$'s open neighborhood contains at least $k$ vertices in $\{x_1, x_2,\ldots,x_{k\cdot (d^{c+1}+d)}\}$, which gives us a $k$-half-quasi-ladder ending in $\{x_1, x_2,\ldots,x_{k\cdot (d^{c+1}+d)}\}$.
\end{proof}

\begin{lemma}\label{final step 2}
Let $k'$ = $k \cdot R(k,k)^{R(k,k)+1} + R(k,k)$, and let $G$ be a graph that contains an  $R(k',k')$-creature ($A$, $B$, $\{x_1, x_2,\ldots,x_{R(k',k')}\}$, $\{y_1, y_2, \ldots,y_{R(k',k')}\}$). Then $G[A \cup \{x_1, x_2,\ldots,x_{R(k',k')}\}]$ contains an induced $k$-half-theta, $k$-half-prism, or a $k$-half-quasi-ladder, ending in $\{x_1, x_2,\ldots,x_{R(k',k')}\}$.
\end{lemma}

\begin{proof}
Let $k'$ = $k \cdot R(k,k)^{R(k,k)+1} + R(k,k)$. Assume that $G$ contains a $R(k',k')$-creature ($A$, $B$, $\{x_1, x_2,\ldots,x_{R(k',k')}\}$, $\{y_1, y_2, \ldots,y_{R(k',k')}\}$). Apply Ramsey's Theorem to $\{x_1, x_2,\ldots,x_{R(k',k')}\}$. If Ramsey's Theorem returns a clique of size $k'$ or more then we have that $G[A \cup \{x_1, x_2,\ldots,x_{R(k',k')}\}]$ contains a $k$-half-prism ending in $\{x_1, x_2,\ldots,x_{R(k',k')}\}$, so we can assume that Ramseys theorem returns an independent set of size at least $k'$. By relabeling the $x_i$'s and $y_i$'s if follows that $G$ contains a $k'$-creature ($A$, $B$, $\{x_1, x_2,\ldots,x_{k'}\}$, $\{y_1, y_2, \ldots,y_{k'}\}$) where $\{x_1, x_2,\ldots,x_{k'}\}$ is an independent set.

Let $A'$ be a minimally connected induced subgraph of $G[A]$ such that $\{x_1, x_2,\ldots,x_{k'}\} \subset N(A')$. If $A'$ contains a vertex of degree R$(k,k)$ in $A'$, then by Lemma \ref{high degree} $G[A \cup \{x_1, x_2,\ldots,x_{k'}\}]$ contains a $k$-half-theta ending in $\{x_1, x_2,\ldots,x_{k'}\}$. So we may assume max degree of $A'$ is R$(k,k)$.

If $A'$ contains R$(k,k)^{R(k,k)}$ vertices of degree greater than two, then there is an induced path of $A'$ that contains R$(k,k)$ vertices of degree greater than two by Lemma \ref{finding path with many good vertices}. Then by Lemma \ref{using path with many good vertices} $G[A \cup \{x_1, x_2,\ldots,x_{k'}\}]$ contains a $k$-half-quasi-ladder or a $k$-half-prism ending in $\{x_1, x_2,\ldots,x_{k'}\}$. So we may assume that $A'$ has maximum degree R$(k,k)$ and contains fewer than R$(k,k)^{R(k,k)}$ vertices of degree greater than two. It then follows from Lemma \ref{final step} that $G[A \cup \{x_1, x_2,\ldots,x_{k'}\}]$ contains a $k$-half-quasi-ladder ending in $\{x_1, x_2,\ldots,x_{k'}\}$.
\end{proof}

The next three lemmas show how to clean up a half-quasi-ladder into a half-ladder, half-theta, or theta. Their proofs are similar to those of lemmas \ref{path domination helper} \ref{path domination of nice set}, and \ref{domination or creature/ladder} respectively, although the conclusions we draw from them are somewhat different.
%in particular we use the next three lemmas to clean up a $k'$-half-quasi-ladder into a $k$-half ladder or a $k$-half-theta.

\begin{lemma}\label{one side path domination helper}
Let ($G$, $S$, $P$, $v$) be a tuple where $G$ is a graph, $v \in G$, $S \subset V(G)$, and $P$ is an induced path of $G$ such that ($S \cup \{v\}$) and $V(P)$ are disjoint. Assume $G[V(P) \cup S \cup \{v\}]$ does not have a $k$-half-theta ending in $S$, then there is a set $X \subset S \cup V(P) \cup \{v\}$ of size at most $4k-1$ such that $N(S-N[X]) \cap N(v) \cap V(P) = \emptyset$, and no vertex of $S-N[X]$ is neighbors with $v$.
\end{lemma}

\begin{proof}
Let $G$, $S$, $P$, and $v$ be as in the statement of this lemma. Number the vertices of $P$ 1 through $|V(P)|$ so that the vertex numbered $i$ is neighbors with the vertices numbers $i-1$ and $i+1$. We now consider the following process to build the set desired set $X$ such that $N(S-N[X]) \cap N(v) \cap V(P) = \emptyset$ and $X \subset S \cup V(P) \cup \{v\}$. 

We do the following for the first step of the process. Let $X_1 = \{v\}$, and let $S_1 = \{s : s\in S-N(X_1)$ and $N(s) \cap N(v) \cap V(P) \neq \emptyset\}$ (i.e. $S_1$ is the set of vertices of $S-N(X_1)$ that share a neighbor with $v$ in $P$). Label the vertices of $S_1$ by the lowest numbered vertex it is neighbors with in $V(P) \cap N(v)$. Let $s_1$ be a highest labeled vertex in $S_1$, and let $p_1$ be $s_1$'s lowest numbered neighbor in $N(v) \cap V(P)$. This completes the first step.

For the $i^{th}$ step we do the following. Let $X_i = X_{i-1} \cup \{s_{i-1}, p_{i-1}\}$, and let $S_i = S_{i-1} - N[X_i]$ and label the vertices of $S_i$ by the lowest vertex it sees in $V(P) \cap N(v)$ (the vertices of $S_i$ inherit their labels from their labels in $S_{i-1}$). Let $s_i$ be a highest labeled vertex in $S_i$ and let $p_i$ be $s_i$'s lowest neighbor in $N(v) \cap V(P)$. Note by how we selected $v$, $s_1$, $p_1$, $s_2$, $p_2$, $\ldots$ $s_i$, $p_i$ that $s_a$, $1 \leq a \leq i$, cannot be neighbors with $p_b$ if $a > b$ since $p_b$ would be in $X_a$ and therefore $s_a$ would not be in $S_a$, and $s_a$ cannot have a neighbor with $p_b$ if $a < b$ since that would contradict either $p_a$ being $s_a$'s lowest numbered neighbor in $N(v) \cap P$ or $s_a$ being a highest labeled vertex in $S_a$. Hence, we then have that among these vertices $s_j$ is only neighbors with $p_j$ for $1 \leq j \leq i$, and $v$ is only neighbors with $p_j$ for $1 \leq j \leq i$. $p_{2i}$ could be neighbors with $p_{2i+1}$ and/or $p_{2i-1}$ since they could be consecutive vertices on the path $P$, but $p_{2i}$ cannot be neighbors with $p_{2j}$. It follows that the set $ \{v\} \cup \{p_2, p_4,\ldots,p_{2c}\} \cup \{s_2, s_4,\ldots,p_{2c}\}$, $2c \leq i$, forms a $c$-half-theta in $G[V(P) \cup S \{v\}]$ ending in $S$.

We continue this process until we reach an $S_j$ that is empty. By what we noted in the previous paragraph, this process cannot go past the $2k^{th}$ step if $G[V(P) \cup S \cup \{v\}]$ does not contain a $k$-half-theta ending in $S$. Set $X$ to be $X_j$. Since $S_j$ is empty, it follows $N(S-N[X]) \cap N(v) \cap V(P) = \emptyset$. We also have that no vertex of $S-N[X]$ is neighbors with $v$ since $v \in X$ and $|X| \leq 4k-1$ since $j \leq 2k$ and since the first step adds a single vertex and each step after that only adds two vertices.
\end{proof}

\begin{lemma}\label{one side path domination of nice set}
Let $(G,S,P)$ be a tuple such that $G$ is a graph, $S \subset V(G)$ such that $S$ cannot be dominated by $4k^2$ vertices and $P$ is an induced path disjoint from $S$ that dominates $S$. Assume $G[V(P) \cup S]$ does not contain a $k$-half-theta ending in $S$. Then there exists a subset $S'$ of $S$ of size $k$ such that no vertex of $P$ has more than one neighbor in $S'$.
\end{lemma}

\begin{proof}

Let $G$, $S$, and $P$ be as in the statement of the lemma. Assume that we have an independent set of vertices vertices $S_{i-1}$ of size $i-1$, $i \leq k$, and a set $Z_{i-1}$ of size at most $4k(i-1)$, with the properties that no vertex $S-N[Z_{i-1}]$ is neighbors with a vertex in $S_{i-1}$, and any vertex in $P$ that is neighbor with some vertex in $S_{i-1}$ has no other neighbors in $S_{i-1}$ nor in $S-N[Z_{i-1}]$. We will use this to produce a set $S_i$ of size $i$ and $Z_i$ of size at most $4k^2i$ with the same properties. Note that the empty set satisfies the conditions of $S_0$.

Let $S' = S-N[Z_{i-1}]$. Let $s$ be some vertex in $S'$, since $i \leq k$ and $S$ cannot be dominated by $4k^2$ vertices, such an $s$ must exists. We can then apply Lemma \ref{one side path domination helper} using $(G, S', P, s$) and to get a set $X$ of size at most $4k-1$ such that $(S'-N[X]) \cap N(s) \cap V(P) = \emptyset$ and no vertex of $S'-N[X]$ is neighbors with $s$. We then set $S_i$ = $S_{i-1} \cup \{s\}$ and $Z_i = Z_{i-1} \cup X$ and we can see these sets satisfies the required properties.

Since the empty set satisfies the properties of $S_0$ and $S$ cannot be dominated by $4k^2$ vertices, we can generate the set $S_k$ which has size $k$ and no vertex of $P$ has more than one neighbor in $S_k$.
\end{proof}

\begin{lemma}\label{one side domination or creature/ladder} 
Let $T$ be an induced $4k[2(4k)^{k+1}]^2$-half-quasi-ladder of a graph $G$ ending in $X$. Assume $T$ does not have an induced $k$-half-theta ending in $X$ and assume that $G$ does not contain an induced $k$-theta. Then $T$ contains a $k$-half ladder ending in $X$.
\end{lemma}

\begin{proof}
Let $G$, $T$, and $X$ be as in the statement of the lemma. Let $P$ be the backbone path of $T$ and $P_1, P_2,\ldots, P_{4k[2(4k)^{k+1}]^2}$ be its auxiliary paths, where the endpoints of $P_i$ are $v_i$ and $x_i$, and the $x_i$'s are the endpoints of $T$, so $x_i \in X$. Let $S$ = $\{v_1, v_2,\ldots, v_{4k[2(4k)^{k+1}]^2}\}$. Clearly, if any vertex of $P$ is neighbors with $k$ distinct $v_i$'s, then $T$ contains a $k$-half-theta ending in $X$. It follows that since $T$ does not have a $k$-half-theta ending in $X$, the vertices of $S$ cannot be dominated by less than $4[2(4k)^{k+1}]^2$ vertices in $T$. Also, if $G[S \cup V(P)]$ contain a $k$-half-theta ending in $S$, then it contains a $k$-half-theta ending in $X$, so we can apply Lemma \ref{one side path domination of nice set} with $(G, P, S)$ to get a set $S' \subset S$ of size $2(4k)^{k+1}$ such that no vertex of $P$ is neighbors with more than one vertex in $S'$. It follows that by only taking the paths $P_i$ such that $v_i \in S'$, that these $P_i$'s together with $P$, form a $2(4k)^{k+1}$-half-quasi-ladder where no vertex of $P$ has a neighbor with more than one vertex in any of the $P_i$'s. We will call this $2(4k)^{k+1}$-half-quasi-ladder $T'$, we will call its backbone path $P'$ so $P' = P$, and we will call the auxiliary paths $P_1', P_2',\ldots, P_{2(4k)^{k+1}}'$ where the endpoints of $P_i'$ are $v_i'$ and $x_i'$, and the $x_i'$'s are the endpoints of $T'$, so $x_i' \in X$. We use $S'$ as before to denote the set of $v_i'$'s.

Now, number the vertices of $P'$ 1 through $|V(P')|$ so that the vertex numbered $i$ is neighbors with the vertices numbers $i-1$ and $i+1$. For a vertex $x$ in $P'$ we will use the notation $n(x)$ to denote the number it has been given in $P'$. For every $s_j \in S'$ let $p_j \in P'$ be the highest numbered neighbor $s_j$ has in $P$. We now set $P_1 = P'$ and $S_1 = S'$. We will consider the following process, where we will try to produce a large independent set in an auxiliary graph related to some $P_i$ and $S_i$ which we will then use to produce a $k$-half-ladder. We will show  this process cannot go past $k$ iterations if $T$ does not have a $k$-half-theta ending in $X$. We will ensure that at the $i^{th}$ step that $V(P_i) \subset V(P')$, $S_i \subset S'$, $|S_i| \geq 2(4k)^{k-i+2}$, $P_i$ is an induced path, and if $s_j \in S_i$ then $p_j \in P_i$. We will also produce induced subpaths $D_i$ of $P$ such that the $D_i$'s are anti-complete with respect to one another and the vertices of $D_i$ will dominate $S_j$ if $i < j$.

At the $i^{th}$ step we do as follows. Create an auxiliary directed graph, $AUX_i$, whose vertex set is $S_i$ and there is an edge from $s_a \in S_i$ to $s_b \in S_i$ if the following condition holds
\begin{enumerate}

    \item $n(p_a) > n(p_b)$ and $s_a$ has  a neighbor $x$ in $P'$ such that $n(x) < n(p_b)$
    
\end{enumerate}
% $n(p_a) < n(p_b)$ and $s_a$ has a neighbor $x$ in $P$ such that $n(x) > n(p_b)$

If the maximum in degree of $AUX_i$ is at most $\frac{1}{4k}|S_i|$ then we stop. If $i \leq k$ (which we will show must happen) then since $|S_i| \geq 2(4k)^{k-i+2}$ this gives an independent set of size at least $k$ by Lemma \ref{directed independent set}. If there is an $s_j \in S_i$ with in degree at least $\frac{1}{4k}|S_i|$ then for at least $\frac{1}{4k}$ fraction of the vertices of $S_i$ must satisfy (1) playing the role of $s_a$ while $s_j$ plays the role of $s_b$. Call this set of vertices $S_{i+1}$. If $s_j \in S_i$ with in degree at least $\frac{1}{4k}|S_i|$ then we do as follows. Define $D_i$ to be the subpath of $P_i$ that is made up of vertices with numbers less than $n(p_j)$. Set $P_{i+1}$ to be the vertices of $P_i$ with numbers greater than $n(p_j)$. This concludes the $i^{th}$ step.

It can then be seen that $V(P_{i+1}) \subset V(P)$, $S_{i+1} \subset S$, $|S_{i+1}| \geq 2(4k)^{k-i+1}$, $P_{i+1}$ is an induced path, and if $s_j \in S_{i+1}$ then $p_j \in P_{i+1}$ as required. Furthermore, it can be seen that any of the previously $D_j$'s that have been produced in this process ($j \leq i$) dominate all vertices of $S_{i+1}$. Since the $D_j$'s are disjoint and anti complete,  By Lemma \ref{many dominating sets} then, this process cannot go past the $k^{th}$ iteration without producing a $k$-theta in $G$.

We conclude there is some step $j \leq k$ such that the auxiliary graph $AUX_j$ has max in-degree less than $\frac{1}{4k}|S_j|$, and since $|S_j| \geq 8k$ it therefore has an independent set of size $k$ by Lemma \ref{directed independent set}. Let $S^*$ denote such an independent set.

 We claim by only taking the paths $P_i'$ such that $v_i' \in S^*$, that these $P_i'$'s together with $P'$, form a $k$-half-ladder. Let $x,y \in S^*$ and let $a,b$ be the highest and lowest numbered neighbors of $x$ in $L$ respectively, and assume that $y$ has a neighbor $c$ on the induced path of $L$ that has $a$ and $b$ as its endpoints. If $y$'s highest numbered neighbor in $L$ is greater than $n(a)$ then $y$ has an edge to $x$ in $AUX_j$. If $y$'s highest numbered neighbor in $L$ is less than $n(a)$, then $x$ has an edge to $y$. It follows that taking the $P_i'$ such that $v_i' \in S^*$ together with $P'$, form a $k$-half-ladder.
\end{proof}

\begin{corollary}\label{final step corollary}
Let $k$ be a natural number. There exists a natural number $k'$ large enough so that if $G$ is be a graph that contains a $k'$-creature ($A$, $B$, $\{x_1, x_2,\ldots,x_{k'}\}$, $\{y_1, y_2, \ldots,y_{k'}\}$), then $G[A \cup \{x_1, x_2,\ldots,x_{k'}\}]$ contains an induced $k$-half-theta, $k$-half-prism, or $k$-half-ladder ending in $\{x_1, x_2,\ldots,x_{k'}\}$ or $G$ contains an induced $k$-theta.
\end{corollary}

\begin{proof}
By Lemma \ref{final step 2} there exists a $k'$ large enough so that if $G$ contains a $k'$-creature ($A$, $B$, $\{x_1, x_2,\ldots,x_{k'}\}$, $\{y_1, y_2, \ldots,y_{k'}\}$) then $G[A \cup \{x_1, x_2,\ldots,x_{k'}\}]$ contains an induced  $4k[2(4k)^{k+1}]^2$-half-theta, $4k[2(4k)^{k+1}]^2$-half-prism, or a $4k[2(4k)^{k+1}]^2$-half-quasi-ladder, ending in $\{x_1, x_2,\ldots,x_{k'}\}$. If $G[A \cup \{x_1, x_2,\ldots,x_{k'}\}]$ contains a $4k[2(4k)^{k+1}]^2$-half-theta or a $4k[2(4k)^{k+1}]^2$-half-prism ending in $\{x_1, x_2,\ldots,x_{k'}\}$ then we are done. If $G[A \cup \{x_1, x_2,\ldots,x_{k'}\}]$ contains a $4k[2(4k)^{k+1}]^2$-half-quasi-ladder ending in $\{x_1, x_2,\ldots,x_{k'}\}$ then we may apply Lemma \ref{one side domination or creature/ladder} to get that either $G[A \cup \{x_1, x_2,\ldots,x_{k'}\}]$ contains a $k$-half-ladder ending in $\{x_1, x_2,\ldots,x_{k'}\}$ or $G$ contains a $k$-theta.
\end{proof}

\begin{lemma}\label{k creature implies theta}   
Let $k$ be a natural number. Then there exists a natural number $k'$ large enough so that if $G$ is a graph that contains an $k'$-creature ($A$, $B$, $\{x_1, x_2,\ldots,x_{k'}\}$, $\{y_1, y_2, \ldots,y_{k'}\}$), then $G$ contains an induced $k$-theta, $k$-prism, $k$-pyramid, $k$-ladder-theta, $k$-ladder-prism, or a $k$-ladder.
\end{lemma}

\begin{proof}
Let $k$ be a natural number. By Corollary \ref{final step corollary} there exists a $k'$ large enough so that if $G$ is a graph that contains an $k'$-creature ($A$, $B$, $\{x_1, x_2,\ldots,x_{k'}\}$, $\{y_1, y_2, \ldots,y_{k'}\}$), then $G[A \cup \{x_1, x_2,\ldots,x_{k'}\}]$ contains an induced $k$-half-theta, $k$-half-prism, or $k$-half-ladder ending in $\{x_1, x_2,\ldots,x_{k'}\}$ or $G$ contains an induced $k$-theta. It then also follows from Corollary \ref{final step corollary} there exists a $k''$ large enough so that if $G$ is a graph that contains an $k''$-creature ($A$, $B$, $\{x_1, x_2,\ldots,x_{k''}\}$, $\{y_1, y_2, \ldots,y_{k''}\}$), then $G[B \cup \{y_1, y_2,\ldots,y_{k''}\}]$ contains an induced $k'$-half-theta, $k'$-half-prism, or $k'$-half-quasi-ladder ending in $\{y_1, y_2,\ldots,y_{k''}\}$ or $G$ contains an induced $k'$-theta.

So, assume that $G$ is a graph that contains an $k''$-creature ($A$, $B$, $\{x_1, x_2,\ldots,x_{k''}\}$, $\{y_1, y_2, \ldots,y_{k''}\}$). If $G$ contains an induced $k'$-theta then we are done, assume that $G[B \cup \{y_1, y_2,\ldots,y_{k''}\}]$ contains an induced $k'$-half-theta, $k'$-half-prism, or $k'$-half-ladder ending in $\{y_1, y_2,\ldots,y_{k''}\}$. By relabeling the $x_i$'s and $y_i$'s we can then assume that $G$ contains a $k'$ creature ($A'$, $B'$, $\{x_1, x_2,\ldots,x_{k'}\}$, $\{y_1, y_2, \ldots,y_{k'}\}$) such that $G[B' \cup \{y_1, y_2,\ldots,y_{k'}\}]$ is a $k'$-half-theta, $k'$-half-prism, or $k'$-half-ladder. Then applying Corollary \ref{final step corollary} gives us that $G[A \cup \{x_1, x_2,\ldots,x_{k'}\}]$ contains an induced $k$-half-theta, $k$-half-prism, or $k$-half-ladder ending in $\{x_1, x_2,\ldots,x_{k}\}$. 
%If $G[B' \cup \{y_1, y_2,\ldots,y_{k'}\}]$ is a $k'$-half-ladder and $G[A \cup \{x_1, x_2,\ldots,x_{k'}\}]$ contains a $10k^2$-half-ladder then we can see that $G$ contains an induced $k$-claw or an induced $k$-paw \todo{too hand wavey? make into a lemma?}. If $G[A \cup \{x_1, x_2,\ldots,x_{k'}\}]$ does not contain a $10k^2$-half-ladder then 
It follows that $G$ must contain a $k$-theta, a $k$-prism, $k$-pyramid, $k$-ladder-theta, $k$-ladder-prism, or a $k$-ladder.
\end{proof}

The following two lemmas will be used in Lemma \ref{exponentially many mins seps} to establish that if ${\cal F}$ is a family of graphs defined by a finite number of forbidden induced subgraphs and ${\cal F}$ allows for at least one of $k$-thetas, $k$-prisms, $k$-pyramids, $k$-ladder-thetas, or $k$-ladder-prisms, for arbitrarily large $k$, then we can ensure it contains these graphs where their number of vertices only grow linearly with respect to $k$, and therefore have exponentially many minimal separators. These two lemmas achieve this by showing that a graph in ${\cal F}$ has certain paths that are too long, then we can contract part of those paths and maintain that the resulting graph remains in ${\cal F}$.

%$k$-claws, $k$-paws for all $k$ larger than some fixed constant, then there exists a constant $c > 1$ such that for every integer $N$ there is a graph $G \in {\cal F}$ such that $|V(G)| = n > N$ and $G$ has at least $c^n$ minimal separators.

\begin{lemma}\label{short paths}
Let $G$ be a graph and let $H$ be a graph with $|V(H)| \leq h$, where $h > 5$. Assume that $G$ contains an induced path $P$ of length at least $5h$ where all internal vertices of $P$ have degree 2 in $G$. Then there exists an edge $e$ in $G$ such that if $G^e$ contains $H$ as an induced subgraph, then so does $G$. 
\end{lemma}

\begin{proof}
Let $G$ be a graph, let $H$ be a graph with $|V(H)| \leq h$ where $h > 5$, and let $P$ be an induced path of $G$ of length at least $5h$ where all internal vertices of $P$ have degree 2, say $P$ = $p_1, p_2,\ldots, p_{5h}$. Let $e$ be the edge between $p_{\lceil \frac{5h-1}{2} \rceil}$ and $p_{\lceil \frac{5h+1}{2} \rceil}$. Let $v$ denote the new vertex $p_{\lceil \frac{5h-1}{2} \rceil}$ and $p_{\lceil \frac{5h+1}{2} \rceil}$ create when $e$ is contracted in $G$ to make $G^e$, and let $P'$ be what the path $P$ becomes after contracting $e$ in $G$, so $P' = p_1, p_2,\ldots, p_{\lceil \frac{5h-1}{2} \rceil -1}, v, p_{\lceil \frac{5h+1}{2} \rceil + 1}, \ldots, p_{5k}$. Assume that $G^e$ contains $H$ as an induced subgraph. We will show that there exists a set $X \subset V(G^e)$ that induces $H$ such that $v \notin X$. It will then follows that $G$ contains an induced $H$.

Any component of $H$ that is not an induced path can only contain vertices outside of $P'$ or within distance $h$ of either the endpoints of $P'$ since all internal vertices of $P'$ have degree 2 in $G^e$. For the components of $H$ that are paths, since there are at most $h$ vertices among these components, we can ensure that the vertices of $X$ that we use to induce these components either do not belong to $P'$ or only contain vertices from the subpaths $p_{h+2}, p_{h+3}, \ldots, p_{\lceil \frac{5h-1}{2} \rceil -1}$ and $p_{\lceil \frac{5h+1}{2} \rceil + 1}, p_{\lceil \frac{5h+1}{2} \rceil}, \ldots, p_{4h-2}$. It follows that $v \notin X$.
\end{proof}

\begin{lemma}\label{small neighborhoods}
Let $G$ be a graph and let $H$ be a graph with $|V(H)| \leq h$, where $h > 5$. Assume that $G$ contains an induced path $P$ of length $5h[(h+1)(5h)^{2h+2} + 1]$ such that the only neighbor the vertices of $P$ might have outside of $P$ is a single vertex $v$. Then there exists a subpath path $P'$ of $P$ such that if $G^{P'}$ contains $H$ as an induced subgraph, then so does $G$. 
\end{lemma}

\begin{proof}
Let $G$ be a graph and let $H$ be a graph with $|V(H)| \leq h$, where $h > 5$. Assume that $G$ contains an induced path $P$ of length $5h[(h+1)(5h)^{2h+2} + 1]$ such that the only neighbor the vertices of $P$ might have outside of $P$ is a single vertex $v$. Let $a,b$ be the endpoints of $P$. Now divide $P$ into a sequence of subpaths $P_1, P_2, \ldots, P_k$ each of length at least 2 so that all internal vertices of $P_i$ have degree 2 in $G$, all endpoints of $P_i$ are either a vertex of degree 3 or $a$ or $b$, $P_1$ has $a$ and one of its endpoints, $P_k$ has $b$ as one of its endpoints, and $P_i$ shares one of its endpoints with $P_{i+1}$ (i.e. these are subpaths that whose endpoints are $a, b$, or the vertices that are neighbors with $v$ and are sequenced going from one end of $P$ to the other). We define a second sequence $a_1, a_2, \ldots a_k$ where $a_i$ = $|E(P_i)|$. If any $a_i \geq 5h$ then the result follows from Lemma \ref{short paths}, so we can assume for all $i$ that $a_i \leq 5h$. It then follows that $k$ is at least $(h+1)(5h)^{2h+2} + 1$, and therefore by the pigeonhole principle there must be a continuous subsequence of length $2h+2$ that is repeated at least $h+2$ times, where none of these continuous subsequences overlap with each other. Let $S$ = $s_0, s_1, \ldots, s_{2h+1}$ be this repeated subsequence. So we have $h+2$ sequences for $1 \leq i \leq h+2$, $A_i = a_{j_i}, a_{j_i+1}, \ldots, a_{j_i+2h+1}$ where for $1 \leq m \leq h+2$ and $c$, $0 \leq c \leq 2h+1$, $a_{j_m+c} = s_{c}$ and no part of $A_m$ overlaps with some other $A_n$  (so $|j_n - j_m| \geq 2h+2$) and $j_m > j_n$ if $m > n$. Fix the values denoted by $j_m$ for $1 \leq m \leq h+2$.

We wish to combine the first half of $A_1$ with the second half of $A_2$ by contracting a path in $P$. Let $x$ be the endpoint of $P_{j_1+h+1}$ that it shares with $P_{j_1+h}$, and let $y$ be the endpoint $P_{j_2+h+1}$ shares with $P_{j_2+h}$. Let $P'$ be the subpath of $P$ that has $x$ and $y$ as its endpoints. Let $w$ be the vertex that gets created when contracting the path $P'$ in $G$ to get $G^{P'}$ and let all the subpaths $P_i$ of $P$ in $G$ that were not contained in $P'$ retain their labels in $G^{P'}$, so $P_{j_1+h}$ and $P_{j_2+h+1}$ share $w$ as an endpoint, and let the $a_i$'s retain their same meaning as long as $P_i$ was not a subpath of $P'$. It follows that $G^{P'}$ has $h$ sequences for $3 \leq i \leq h+2$, $A_i = a_{j_i}, a_{j_i+1}, \ldots, a_{j_i+2h+1}$ where for $1 \leq m \leq h+2$ and $c$, $0 \leq c \leq 2h+1$, $a_{j_m+c} = s_{c}$ and no part of $A_m$ overlaps with some other $A_n$  (so $|j_n - j_m| \geq 2h+2$) and $j_m > j_n$ if $m > n$. Furthermore, $A_1$ and $A_2$ have now been combined to give $A' = a_{j_1}, a_{j_1+1}, \ldots, a_{j_1+h}, a_{j_2+h+1}, a_{j_2+h+2}, \ldots, a_{j_2+2h+1}$ so that $a_{j_1+c} = s_{c}$ for $0 \leq c \leq h$ and $a_{j_2+c} = s_{c}$ for $h+1 \leq c \leq 2h+1$.
%, and the paths $P_{j_1+h}$ and $P_{j_2+h+1}$ share $w$ as their common endpoint. 
We will show that if there exists a set $X \subset V(G^{P'})$ that induces $H$ in $G^{P'}$ then we can require $w \notin X$. The result then follows since if $w \notin X$ then the vertices that correspond to $X$ in $G$ induced an $H$ in $G$.

So, assume $X \subset V(G^{P'})$ and induces $H$. If $w \notin X$ then we are done, so assume $w \in X'$ for some connected component $X'$ of $X$. For $i$ with $3 \leq i \leq h+1$, let $P_i^*$ denote the path induced by $V(P_{j_i}), V(P_{j_i+1}), \ldots, V(P_{j_i+2h+1})$ in $G^{P'}$, so $P_i^*$ is the path that naturally corresponds to $S_i$, and let $P_1^*$ denote the path induced by $V(P_{j_1}), V(P_{j_1+1}), \ldots, V(P_{j_1+h}), V(P_{j_2+h+1}), V(P_{j_2+h+2}), \ldots, V(P_{j_2+2h+1})$, so $P_1^*$ naturally corresponds with $A'$.  
%If $w \in X$ then we are done, so assume there is some connected component $X'$ in $X$ such that $v \in X'$.
Then since $X'$ has at most $h$ vertices there is at least one $P_i^*$ that contains no vertex of $X$ and since $X'$ is connected and contains $w$, all vertices of $X' \cap P$ must be completely contained in $V(P_1^*)$ since $w$ is at least distance $h$ from either endpoint of $P^*_1$. It follows that we can replace the vertices of $X' \cap P$, which must be completely contained in the interal vertices of $P_1^*$, with the corresponding vertices in a $P_i^*$ that contains no vertices of $X$ and still maintain that the vertices of $X$ induce $H$. Now $w \notin X$ and the result then follows.
\end{proof}

\begin{lemma}\label{exponentially many mins seps}
Let ${\cal F}$ be a family of graphs determined by a finite number of forbidden induced subgraphs. Then if ${\cal F}$ does not forbid all $k$-thetas, $k$-prisms, $k$-pyramids, $k$-ladder-thetas, $k$-ladder-prisms, and $k$-ladders for arbitrarily large $k$, then ${\cal F}$ is feral. 
%%%%%%%%%%%%%%%%%%%%%%%%%%%%%%%%%%%%%%%%%%%%%%%%there exists a constant $c > 1$ such that for every natural number $N$ there exists a $G \in {\cal F}$ such that $|V(G)|=n > N$ and $G$ has at least $c^n$ minimal separators.
\end{lemma}

\begin{proof}
Let ${\cal F}$ be a family of graphs determined by a finite number of forbidden induced subgraphs, and let ${\cal H}$ be a set of forbidden subgraphs that define ${\cal F}$. Let let $h > 5$ be a number such that for any $H \in {\cal H}$, $|V(H)| \leq h$. First assume that ${\cal F}$ allows for either $k$-thetas $k$-prisms, or $k$-pyramids for arbitrarily large $k$. Then by Lemma \ref{short paths} we can ensure that all paths with internal vertices all having degree 2 of the $k$-thetas $k$-prisms, or $k$-pyramids are at most $5h$ (we keep on contracting the appropriate edges given by Lemma \ref{short paths} until no path where all internal vertices have degree 2 have length more than $5h$) and therefore ${\cal F}$ contains a $k$-theta $k$-prism, or $k$-pyramid with at most $5h \cdot k$ vertices. Since a $k$-theta, $k$-prism, or $k$-pyramid must have at least $2^k$ minimal separators, it follows that there exists a $c > 1$ such that for every natural number $N$ ther exists a $G \in {\cal F}$ such that $|V(G)| = n > N$ and the number of minimal separators in $G$ is at least $c^n$.

Now assume that ${\cal F}$ allows for $k$-ladder-thetas or $k$-ladder-prisms for arbitrarily large $k$. Every $k$-ladder-theta and $k$-ladder-prism contains a $k$-half-ladder and by Lemma \ref{short paths} we can ensure that all paths with internal vertices all having degree 2 of the $k$-ladder-theta or $k$-ladder-prism are at most $5h$ and by Lemma \ref{small neighborhoods} we can ensure that the backbone path of the corresponding $k$-half-ladder has length at most $[5h(h+1)(5h)^{2h+1} + 1] \cdot k$ by contracting the appropriate edges and paths if necessary while still guaranteeing the resulting graph belongs to ${\cal F}$ (Lemma \ref{small neighborhoods} gives us that if there is a subpath of length over $[5h(h+1)(5h)^{2h+1} + 1]$ of the backbone path that only has one neighbor outside of the backbone path, there there exists a subpath of the backbone path that we can contract and still maintain that the resulting graph is a $k$-ladder-theta or $k$-ladder-prism contained in ${\cal F}$). Since $k$-ladder-thetas and $k$-ladder-prisms have at least $2^k$ minimal separators it follows that there exists a contains $c > 1$ such that for every natural number $N$ there exists a $G \in {\cal F}$ such that the number of minimal separators in $G$ is at least $c^n$. It follows that ${\cal F}$ is feral.
\end{proof}

The following lemma shows why it is necessary to forbid $k$-paw and $k$-claw graphs for a family of graphs defined by a finite number of forbidden induces subgraphs to be strongly-quasi-tame. Figure \ref{paw claws graph} gives a picture of the two graphs constructed in the following lemma.

\begin{lemma}\label{k-claws are untame}
Let ${\cal F}$ be a family of graphs determined by a finite number of forbidden induced subgraphs. Then if ${\cal F}$ does not forbid $k$-claws and $k$-paws for some natural number $k$, then ${\cal F}$ is feral.
%%%%%%%%%%%%%%%%%%%%%%%%%%%%%%%%then there exists a $c > 1$ such that for every natural number $N$ there exists a $G \in {\cal F}$ such that $|V(G)| = n > N$ and the number of minimal separators in $G$ is at least $c^n$.
\end{lemma}

\begin{proof}
Let ${\cal F}$ be a family of graphs determined by a finite number of forbidden induced subgraphs, and let ${\cal H}$ be a set of forbidden subgraphs that define ${\cal F}$. Let $h > 5$ be a number such that for any $H \in {\cal H}$, $|V(H)| \leq h$. First we assume that ${\cal F}$ allows $k$-claw for arbitrarily large $k$. We will construct a graph with many minimal separators. Assume that we have two set of  $2^c-1$ long-claws, $C^1_1, C^1_2, \ldots C^1_{2^c}$, and $C^2_1, C^2_2, \ldots C^2_{2^c}$ where in both sets each long claw has arm length $h$. We label the leaves of $C^1_i$ as $a^1_i, b^1_i, c^1_i$ and we label the endpoints of $C^2_i$ as $a^2_i, b^2_i, c^2_i$. Then for $1 \leq i \leq 2^{c-1}-1$ we glue $a^1_{2i}$ to $b^1_i$, $a^1_{2i+1}$ to $c^1_i$, $a^2_{2i}$ to $b^2_i$, and $a^2_{2i+1}$ to $c^2_i$. Furthermore, for $2^{c-1} \leq i \leq 2^c-1$ we add an edge between $b^1_i$ and $b^2_i$ and between $c^1_i$ and $c^2_i$. Note that any collection of $b_i^{j_i}$ and $c_i^{\ell_i}$ with $2^{c-1} \leq i \leq 2^c-1$ and $j_i, \ell_i$ = 1 or 2 is a minimal separator, so there are at least $2^{2^c}$ minimal separators in this construction. Since the arm length of each long-claw is $h$, the total number of vertices in this construction is less than $3h \cdot 2^{c+1}$.

If ${\cal F}$ allows for $k$-claws, then forest of paths and subdivided claws cannot be forbidden in ${\cal F}$, and it can be seen that any induced subgraph of size at most $h$ of the construction just given is a forest of paths and subdivided claws (i.e. three anti-complete paths where one endpoint of each path are glued together). It follows that this construction must belong to ${\cal F}$ and since this construction has at least $2^{2^c}$ minimal separators and less than $3h \cdot 2^{c+1}$ vertices, the statement of the lemma follow for the case where $k$-claw graphs for arbitrarily large $k$ are not forbidden.

Now we assume  that ${\cal F}$ allows $k$-paw graphs for arbitrarily large $k$. The construction and analysis we make in this case is nearly identical to the $k$-claw case. We present it here for completeness. Assume that we have two set of  $2^c-1$ long-paws, $C^1_1, C^1_2, \ldots C^1_{2^c}$, and $C^2_1, C^2_2, \ldots C^2_{2^c}$ where in both sets each long-paw has arm length $h$. We label the endpoints of $C^1_i$ as $a^1_i, b^1_i, c^1_i$ and we label the endpoints of $C^2_i$ as $a^2_i, b^2_i, c^2_i$. Then for $1 \leq i \leq 2^{c-1}-1$ we glue $a^1_{2i}$ to $b^1_i$, $a^1_{2i+1}$ to $c^1_i$, $a^2_{2i}$ to $b^2_i$, and  $a^2_{2i+1}$ to $c^2_i$. Lastly, for $2^{c-1} \leq i \leq 2^c-1$ we add an edge between $a^1_i$ and $a^2_i$ and between $b^1_i$ and $b^2_i$. Note that any collection of $b_i^{j_i}$ and $c_i^{\ell_i}$ with  $2^{c-1} \leq i \leq 2^c-1$ and $j_i, \ell_i$ = 1 or 2 is a minimal separator, so there are at least $2^{2^c}$ minimal separators in this construction. Since the arm length of each long-claw is $h$, the total number of vertices in this construction is less than $3h \cdot 2^{c+1}$.

Since ${\cal F}$ allows for $k$-paws, a forest of paths and subdivided paws cannot be forbidden in ${\cal F}$, and it can be seen that any induced subgraph of size at most $h$ of the construction just given is a forest of paths and subdivided paws. It follows that this construction must belong to ${\cal F}$ and since this construction has at least $2^{2^c}$ minimal separators and less than $3h \cdot 2^{c+1}$ vertices, the statement of the lemma follows for the case where $k$-paw graphs for arbitrarily large $k$ are not forbidden.
\end{proof}

We are now ready to prove Theorem \ref{theorem 2}

\begin{figure}
\centerline{\includegraphics[height = 4.5cm, width=1\linewidth, scale=.8]{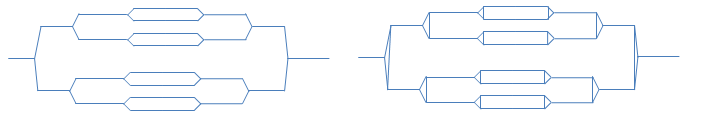}}
\caption{\em The two graphs in this figure are small versions of the constructions of the graphs given in Lemma \ref{k-claws are untame}, explicit vertices are omitted in this graph. The left side graph is the construction provided when when the $k$-claw is not forbidden for arbitrarily large $k$. The right hand side graph is the construction provided when when the $k$-paw is not forbidden for arbitrarily large $k$.}
\label{paw claws graph}
\end{figure}%\todo{something is funny looking with this figure}

\begin{proof}[Proof of Theorem~\ref{theorem 2}]
Let ${\cal F}$ be a family of graphs defined by a finite number of forbidden induced subgraphs. It follows from Lemmas \ref{exponentially many mins seps} and \ref{k-claws are untame} that if ${\cal F}$ allows for any $k$-thetas, $k$-prisms, $k$-pyramids, $k$-ladder-thetas, $k$-ladder-prisms, $k$-claws, or $k$-paws for arbitrarily large $k$, ${\cal F}$ is feral.
%%%%%%%%%%%%%%%%%%%%%%%%%%%%%%%%%%%%%%%%%%%%%%%then there exists a constant $c > 1$ such that for every natural number $N$ there exists a $G \in {\cal F}$ such that $|V(G)| = n > N$ and $G$ has at least $c^n$ minimal separators.
%does not exists a quasi-polynomial function $P(n)$ such that for every $G \in {\cal F}$ the number of minimal separators in $G$ is at most $P(|V(G)|)$. 
%So, if there does exists a function $f : \mathbb{N} \rightarrow \mathbb{N}$ such that for all $G \in {\cal F}$ the number of minimal separators of $G$ is at most $n^{f(k)\log(n)}$ it follows that for all $k$ larger than some fixed integer, ${\cal F}$ must forbid all $k$-thetas, $k$-prisms, $k$-pyramids, $k$-ladder-thetas, $k$-ladder-prisms, $k$-claws, and $k$-paws.

Now assume that there exists a natural number $k$ such that ${\cal F}$ forbids $k$-thetas, $k$-prisms, $k$-pyramids, $k$-ladder-thetas, $k$-ladder-prisms, $k$-claws, and $k$-paws. Observe that there exists a $k'$ large enough so that if $G$ contains an induced $k'$-ladder, then $G$ contains an induced $k$-claw or $k$-paw graph, therefore ${\cal F}$ forbids $k'$-ladders. It then follows from Lemma \ref{k creature implies theta} there exists a $k''$ such that no $G \in {\cal F}$ can contain a $k''$-creature, where the minimum value of $k''$ is a function of $k$. Furthermore, it is clear that there exists a $k'''$ large enough so that if $G$ contains a $k'''$-skinny-ladder as an induced minor, then $G$ contains a $k$-claw or a $k$-paw as an induced subgraph. Hence ${\cal F}$ forbids $k'''$-skinny-ladders as an induced minor. It then follows from Theorem \ref{theorem 1} that there is a function $f : \mathbb{N} \rightarrow \mathbb{N}$ such that for all $G \in {\cal F}$ the number of minimal separators of $G$ is at most $n^{f(k)\log(n)}$. Hence ${\cal F}$ is tame.
%quasi-polynomial function $P(n)$ such that for every $G \in {\cal F}$ the number of minimal separators in $G$ is at most $P(|V(G)|)$.
\end{proof}

\section{Long Cycle-free Graphs}\label{cycle free graphs}

Here we present a proof of Theorem \ref{theorem 3} which is based on an easy application of Corollary \ref{outside domination}. We will need the following lemma in order to apply Corollary \ref{outside domination}.

\begin{lemma}\label{cycle free domination}
Let $G$ be a $C_{\geq r}$-free graph and assume $G$ does not contain a $k$-creature. Then every minimal separator, $S$, can be dominated by $r\cdot k^2$ vertices of $G$ not in $S$.
\end{lemma}

\begin{proof}
Let $G$ be a $C_{\geq r}$-free graph and assume $G$ does not contain a $k$-creature. Assume for a contradiction that there exists a minimal separator, $S$, of $G$ such that $S$ cannot be dominated by $r\cdot k^2$ vertices in $G$ and not in $S$. Let $H$ be an $S$-full component of $G-S$, then by Lemma \ref{path domination}, $S$ is dominated a subset of $H$ that is the union of $k^2$ induced paths in $H$. It follows there must exists some induced path $P$ in $H$ such that $S_P = N(P) \cap S$ cannot be dominated by $r$ vertices in $P$. There then exists a subpath $P'$ of $P$ such that there are vertices $a,b \in S_P$ that have no neighbor in $P'$, both component of $P-P'$ have vertices that are neighbors with $a$ and/or $b$. It follows that we can extend the path $P'$ to have endpoints $x_a$ and $x_b$
%Hence there must be a set of vertices, $X$ of $P$ of size at least $r+1$ such that every vertex in $X$ has a neighbor in $P_S$ that no other vertex of $X$ is neighbors with. 
%This implies there must exists vertices $a, b \in S_P$ and a subpath $P'$ of $P$ of length at least $r$ with endpoints $x_a$ and $x_b$
such that the only neighbors of $a$ in $P'$ is $x_a$ and possible $x_b$ and the only neighbors of $b$ in $P'$ is $x_b$ and possibly $x_a$. If $x_a$ and $x_b$ are both neighbors with $a$ then $P'$ and $a$ form a cycle of length $r$, and if $x_a$ and $x_b$ are both neighbors with $b$ then $P'$ and $b$ form a cycle of length $r$ so assume neither of these cases occur. If $a$ and $b$ are neighbors then $P'$ $a$, $b$ make a cycle of length more than $r$. Else, there is an induced path, $T$ between $a$ and $b$ with all of its internal vertices contained in some $S$-full component other than $H$. It follows that $P'$, and $T$ makes a cycle of length more than $r$, a contradiction.
\end{proof}

\begin{proof}[Proof of Theorem~\ref{theorem 3}]
Let $G$ be a $C_{\geq k}$-free graph that is $k$-theta, $k$-prism, and $k$-pyramid free. Since $G$ is $C_{\geq k}$-free this implies that $G$ is also $k$-ladder-theta, $k$-ladder-prism, and $k$-ladder free. Lemma \ref{k creature implies theta} then implies that there exists a function $f : \mathbb{N} \rightarrow \mathbb{N}$ (independent of the choice of $k$ or $G$) such that $G$ is $f(k)$-creature-free. Lemma \ref{cycle free domination} gives that every minimal separator $S$ of $G$ can be dominated by $kf(k)^2$ vertices not in $S$. Hence, by Corollary \ref{outside domination} $G$ has at most $|V(G)|^{(kf(k)^2)^2+2kf(k)^2}$ minimal separators. It follows that the family of graphs that are $C_{\geq k}$-free, $k$-theta, $k$-prism, and $k$-pyramid free is tame.
%depending only of the value of $r$ and $k$. implies that $G$ is $k$ Let $c = $R$^{r+1}(2k^2,2k^2)$, then it follows from Lemma \ref{full theta} that $G$ is $R^{r+1}(2c^2, 2c^2)$-creature free. Lemma \ref{cycle free domination} then gives every minimal separator, $S$, of $G$ can be dominated by $r(R^{r+1}(2c^2, 2c^2))^2$ vertices in $G-S$. It then follows from Corollary \ref{outside domination} that there exists a function  $f : \mathbb{N} \times \mathbb{N} \rightarrow \mathbb{N}$ such that then $G$ has $|V(G)|^{f(k,r)}$ minimal separators.
\end{proof}

\section{Graph With Bounded Clique Size}\label{bounded clique}
Here we present a proof of Theorems \ref{theorem 4} and \ref{theorem 5} which are based on an easy application of Corollary \ref{outside domination}. We will need the following lemma in order to apply Corollary \ref{outside domination}.

\begin{lemma}\label{bounded clique outside domination}
Let $k'$ = $4[(8k^2)^{k+1}]^7$. 
%There exists a function $f : \mathbb{N} \rightarrow \mathbb{N}$ such that 
If $G$ is $k$-creature free, $G$ does not contain a $k$-skinny-ladder as an induced minor, and no minimal separator of $G$ contains a clique of size $k$, then every minimal separator $S$ of $G$ can be dominated by at most $(k')^{k+1}$ vertices of $G-S$. 
%can be dominated by $k$ vertices and no minimal separator contains a clique of size $r$, then every minimal separator $S$ of $G$ can be dominated by at most $f(k,r)$ vertices not contained $S$. 
\end{lemma}

\begin{proof}
Let $k'$, $k$, and $G$ be as in the statement of the lemma. Let $G'$ be an induced subgraph of $G$ and let $S'$ be a minimal separator of $G'$. Then $G'$ 
%Let $S$ be a minimal separator of $G$, and let $S' \subset S$. Then $S'$ is a minimal separator of $(G-S) \cup S'$ and since $(G-S) \cup S'$ 
must be $k$-creature free and $k$-ladder free, so it follow from Lemma \ref{domination or creature/ladder} that $S'$ can be dominated by $k'$ vertices of $G'-S'$.

We will produce a set of $(k')^{k+1}$ vertices of $G-S$ that dominate $S$ by considering the following recursive algorithm. The input to the algorithm is ($G'$,$S'$) where $G'$ is a subgraph of $G$ and $S'$ is a minimal separator of $G'$, and the algorithm returns a set of vertices which will be described shortly. The algorithm finds two vertex sets $A$ and $B$ such that $|A| + |B| \leq k'$,  $A \subset V(G')$, $B \subset S'$, and $A \cup B$ dominate $S'$ (such a set must exists by what was established in the previous paragraph). Let $B'$ be a set of vertices in $G'-S'$ such that $|B'| \leq |B|$ and $B'$ dominates $B$. For each $b \in B$ we recursively call the algorithm on ($G'-(S'-[S' \cap N(b)]), S' \cap N(b))$ (note that $S' \cap N(b)$ is a minimal separator of $G'-(S'-[S' \cap N(b)])$). Let $X$ be the union of the sets returned by each recursive call. Then algorithm then returns $X \cup A \cup B'$.

If we initially call this algorithm on $(G,S)$ for some minimal separator $S$ of $G$, then it is clear that the set this algorithm returns is a subset of vertices of $G-S$ that dominate $S$. We can also see the depth of this recursive algorithm cannot go past $k$ without producing a clique of size $k$ in $S$ since the minimal separator we recursively call this algorithm on is always dominated by the open neighborhood of some vertex $v$ of $S$. So, the depth of the recursion tree is at most $k-1$ and each node has at most $k'$ children since $|B| \leq k'$. It follows that since each recursive call of the algorithm adds at most $k'$ vertices to the set it returns, the size of the final returned set cannot exceed $k'\cdot k'^{k}$
\end{proof}

\begin{proof}[Proof of Theorem~\ref{theorem 4}]
Let $G$ be a graph that is $k$-creature free and does not contain a $k$-skinny-ladder as an induced minor, and furthermore assume that no minimal separator of $G$ has a clique of size $k$. By Lemma \ref{bounded clique outside domination} there exists a function $f : \mathbb{N} \rightarrow \mathbb{N}$ such that all minimal separators, $S$, of any graph that is $k$-creature free, does not contain a $k$-skinny-ladder as an induced minor, and has no minimal separator that contains a clique of size $k$, can be bounded by $f(k)$ vertices outside of $S$. It then follows from Corollary \ref{outside domination} that $G$ has at most $|V(G)|^{f(k)^2+2f(k)}$ minimal separators. Hence, the family of graphs that are $k$-creature free, do not contain a $k$-skinny-ladder as an induced minor, and have no minimal separator has a clique of size $k$ is tame.
\end{proof}

\begin{proof}[Proof of Theorem~\ref{theorem 5}]
Let ${\cal F}$ be a family of graphs defined by a finite number of forbidden induced subgraphs. Assume that ${\cal F}$ forbids the complete graph on $k$ vertices for some natural number $k$. It follows from Lemmas \ref{exponentially many mins seps} and \ref{k-claws are untame} that if ${\cal F}$ allows for any $k'$-thetas, $k'$-ladder-thetas, $k'$-claws, or $k'$-paws for arbitrarily large $k'$, then ${\cal F}$ is feral.
%%%%%%%%%%%%%%%%%%%%%%%%%%%%%%%%%%%%%%%%there does not exists a polynomial function $P(n)$ such that for every $G \in {\cal F}$ the number of minimal separators in $G$ is at most $P(|V(G)|)$. So, if there does exists a polynomial function $P(n)$ such that for every $G \in {\cal F}$ the number of minimal separators in $G$ is at most $P(|V(G)|)$ then it follows that for all $k'$ larger than some fixed integer, ${\cal F}$ must forbid all $k$-thetas, $k$-prisms, $k$-pyramids, $k$-ladder-thetas, $k$-ladder-prisms, $k$-claws, and $k$-paws.

Now assume that for some integer $k$ that ${\cal F}$ forbids $k$-thetas, $k$-ladder-thetas, $k$-claws, and $k$-claws. Since ${\cal F}$ forbids $k$-cliques as well, it follows that ${\cal F}$ forbids $k$-prisms, $k$-pyramids, and $k$-ladder-prisms. Observe that there exists a $k'$ large enough so that if $G$ contains an induced $k'$-ladder, then $G$ contains an induced $k$-claw or $k$-paw, therefore $G$ does not contain a $k'$-ladder. It follows from Lemma \ref{k creature implies theta} there exists a $k''$ such that no $G \in {\cal F}$ can contain a $k''$-creature, where the minimum value of $k''$ is a function of $k$. Furthermore, it is clear that there exists a $k'''$ large enough so that if $G$ contains a $k'''$-skinny-ladder as an induced minor, then $G$ contains a $k$-claw or a $k$-paw as an induced subgraph. Hence ${\cal F}$ forbids $k'''$-skinny-ladders as an induced minor. Now, if no graph of ${\cal F}$ contains a minimal separator with a clique of size $k$, then it then follows by Lemma \ref{bounded clique outside domination} there exists a function $f : \mathbb{N} \rightarrow \mathbb{N}$ such that for all $G \in {\cal F}$ it holds that all minimal separators $S$ of $G$ can be bounded by $f(k)$ vertices in $G-S$. It then follows from Corollary \ref{outside domination} that for all $G \in {\cal F}$ has at most $|V(G)|^{f(k)^2+2f(k)}$ minimal separators. Therefore ${\cal F}$ is tame.
\end{proof}

\section{Conclusion}\label{sec:conclusion}
In this paper we disproved a conjecture of Abrishami et al.~\cite{abrishami2020graphs} that for any natural number $k$, the family of graphs that exclude $k$-creatures is tame. On the other hand, we proved a weakened form of the conjecture, that every family of graphs that excludes $k$-creatures and also excludes $k$-skinny ladders as induced minors is strongly-quasi-tame. This led to a complete classification of graph families defined by a finite number of forbidden induced subgraphs into strongly-quasi-tame and feral, substantially generalizing the main result of Milani\v{c} and Piva\v{c}~\cite{milani2019minimal}. 
The tools we develop on the way to prove our main results yield with some additional effort polynomial upper bounds instead of quasi-polynomial, proving tameness instead of strong quasi-tameness, for two interesting special cases. In particular we show that the conjecture of Abrishami et al.~\cite{abrishami2020graphs} is true for $C_{\geq r}$-free graphs for every integer $r$, as well as for $K_r$-free graphs excluding an $r$-skinny ladder for every integer $r$. The first of these results generalizes work of Chudnovsky et al.~\cite{chudnovsky2019maximum}, who proved that $C_{\geq 5}$-free, $k$-creature free graphs are tame,

%the tools we used in the first two theorems we proved Theorem \ref{theorem 3} which showed that for any natural number $k$ the family of graphs that are $C_{\geq k}$-free, $k$-theta-free, $k$-pyramid-free, and $k$-prism-free are tame which generalized results from \cite{chudnovsky2019maximum} who proved an analogous result for $C_{\geq 5}$-free graphs (long-hole-free graphs). Furthermore, we shows in Theorems \ref{theorem 4} and \ref{theorem 5} that if we require the family of graphs we study to also forbid large cliques, then the strongly-quasi-tame bounds of Theorems \ref{theorem 1} and \ref{theorem 2} can be improved to tame.

Although Theorems \ref{theorem 1} and \ref{theorem 2} provide a strongly-quasi-tame bound we have no examples of non-tame families that exclude $k$-creatures and  $k$-skinny ladders for some $k$. We conjecture that these classes of graphs are actually tame. 
\begin{conjecture} \label{our conjecture}
For every natural number $k$, the family of graphs that are $k$-creature free and do not contain a $k$-skinny-ladder as an induced minor is tame.
\end{conjecture}
Conjecture~\ref{our conjecture}, if true, put together with the proof of Theorem~\ref{theorem 2} would lead to the following classification of hereditary families defined by a finite set of forbidden induced subgraphs. 
%
%If conjecture \ref{our conjecture} is true then the next conjecture will follow from the work done in Section \ref{finite forbidden}.
%
\begin{conjecture} \label{our conjecture2}
Let ${\cal F}$ be a graph family defined by a finite number of forbidden induced subgraphs. If there exists a natural number $k$ such that ${\cal F}$ forbids all $k$-theta, $k$-prism, $k$-pyramid, $k$-ladder-theta, $k$-ladder-prism, $k$-claw, and $k$-paw graphs, then ${\cal F}$ is tame. Otherwise ${\cal F}$ is feral.
\end{conjecture}
We remark that Conjecture~\ref{our conjecture} implies Conjecture~\ref{our conjecture2}, but not the other way around. In particular Conjecture~\ref{our conjecture2} might be easier to prove.

We have so far been unsuccessful in identifying other counterexamples to Conjecture~\ref{false conjecture} that look ``substantially different'' from the $k$-twisted ladders constructed in Section~\ref{counter example section}. For this reason it is tempting to conjecture that at least for induced minor closed classes, a "clean" classification of all classes into tame or feral is possible.
\begin{conjecture}\label{con:inducedMinorcoNJ}
Every induced-minor-closed class {\cal F} is either tame or feral. 
\end{conjecture}
Since removing vertices and contracting edges can not increase the number of minimal separators, Conjecture~\ref{con:inducedMinorcoNJ}, would show (in an informal sense) that both the brittleness of the boundary between tame and non-tame hereditary classes, as well as the existence of non-tame hereditary classes that are not feral is primarily due to ``number fiddling'' effects such as in the example of Abrishami et al.~\cite{abrishami2020graphs} of a tame family containing $k$-creatures for arbitrarily large $k$.

\bibliographystyle{alpha}
\bibliography{bibliography}
\end{document}